\documentclass[aoas]{imsart}

\RequirePackage{amsthm,amsmath,amsfonts,amssymb}
\RequirePackage[authoryear]{natbib}
\RequirePackage[colorlinks,citecolor=blue,urlcolor=blue]{hyperref}
\RequirePackage{graphicx}

\RequirePackage{amsthm,amsmath,amsfonts,amssymb}
\RequirePackage[authoryear]{natbib}
\RequirePackage{fancyhdr, graphicx, proof, comment, multicol}
\RequirePackage{subfigure}
\RequirePackage{caption}
\RequirePackage{pgf, tikz}
\RequirePackage[toc,page]{appendix}
\RequirePackage[ruled,vlined]{algorithm2e}

\usepackage{xr}

\makeatletter
\newcommand*{\addFileDependency}[1]{
  \typeout{(#1)}
  \@addtofilelist{#1}
  \IfFileExists{#1}{}{\typeout{No file #1.}}
}
\makeatother

\newcommand*{\myexternaldocument}[1]{%
    \externaldocument{#1}%
    \addFileDependency{#1.tex}%
    \addFileDependency{#1.aux}%
}

\myexternaldocument{Supplement}

\startlocaldefs
\DeclareMathOperator*{\argmax}{arg\,max}
\DeclareMathOperator*{\argmin}{arg\,min}
\DeclareMathOperator*{\logit}{logit}

\DeclareMathOperator*{\essinf}{ess\,inf}
\newtheorem{theorem}{Theorem}[section]
\newtheorem{supptheorem}{Supplemental Theorem}[section]
\newtheorem{lemma}[theorem]{Lemma}
\theoremstyle{remark}
\newtheorem{definition}[theorem]{Definition}

\endlocaldefs

\begin{document}

\renewcommand{\P}{\mbox{$\mathbb{P}$}}
\newcommand{\E}{\mbox{$\mathbb{E}$}}
\newcommand{\R}{\mbox{$\mathbb{R}$}}
\newcommand{\KL}{\mbox{$\mathcal{D}$}}
\newcommand{\Cov}{\mathrm{Cov}}
\newcommand{\V}{\mathrm{Var}}
\newcommand{\vect}[1]{\boldsymbol{#1}}
\newcommand{\norm}[1]{||{#1}||}
\newcommand{\bignorm}[1]{\bigg|\bigg|{#1}\bigg|\bigg|}
\newcommand{\indep}{\perp \!\!\! \perp}

\newcommand{\al}[1]{\begin{align*}#1\end{align*}}
\newcommand{\gary}[1]{{\bf\textcolor{red}{[Gary: #1]}}}
\newcommand{\yc}[1]{{\bf\textcolor{red}{[Yen-Chi: #1]}}}

\newcommand{\optsupII}[4]{\ensuremath{
\begin{equation*}
\begin{aligned}
    & \underset{#1}{\text{sup}}
    & & #2 \\
    & \text{subject to}
    & & #3 \\
    &&& #4
\end{aligned}
\end{equation*}
}}

\usetikzlibrary{arrows, automata}
\let\hat\widehat
\let\tilde\widetilde

\begin{frontmatter}
\title{Data Harmonization Via Regularized Nonparametric Mixing Distribution Estimation}
\runtitle{Data Harmonization}

\begin{aug}
\author[A]{\fnms{Steven} \snm{Wilkins-Reeves}\ead[label=e1]{stevewr@uw.edu}},
\author[A]{\fnms{Yen-Chi} \snm{Chen}\ead[label=e2,mark]{yenchic@uw.edu}}
\and
\author[B]{\fnms{Kwun Chuen Gary} \snm{Chan}\ead[label=e3,mark]{kcgchan@uw.edu}}
\address[A]{Department of Statistics,
University of Washington,
\printead{e1,e2}}

\address[B]{Department of Biotatistics,
University of Washington,
\printead{e3}}
\end{aug}

\begin{abstract}

Data harmonization is the process by which an equivalence is developed between two variables measuring a common trait.  Our problem is motivated by dementia research in which multiple tests are used in practice to measure the same underlying cognitive ability such as language or memory.  We connect this statistical problem to mixing distribution estimation. 
We introduce and study a non-parametric latent trait model, develop a method which enforces uniqueness of the regularized maximum likelihood estimator, show how a nonparametric EM algorithm will converge weakly to its maximizer, and additionally propose a faster algorithm for learning a discretized approximation of the latent distribution.  Furthermore, we develop methods to assess goodness of fit for the mixing likelihood which is an area neglected in most mixing distribution estimation problems.  We apply our method to the National Alzheimer’s Coordination Center Uniform Data Set and show that we can use our method to convert between score measurements and account for the measurement error. We show that this method outperforms standard techniques commonly used in dementia research.

\end{abstract}

\begin{keyword}
\kwd{Alzheimer's Disease}
\kwd{Nonparametric Expectation-Maximization Algorithm}
\kwd{Latent Trait Model}
\kwd{Measurement Error Model}
\end{keyword}

\end{frontmatter}

\tableofcontents
\section{Introduction}
Data harmonization is the process by which measurements from different sources and by different methods are combined into a single dataset for further analysis.  It often involves converting scores between two different measurements of the same trait, which is the main focus of this paper. For example, in neuropsychological testing for Alzheimer's Disease, collections of tests known as batteries are used to measure a variety of cognitive traits such as attention, language, episodic memory, visual-spatial ability etc.  Different studies, or even in the same study, may use different validated instruments to measure each trait.  As an increasing number of data sets are becoming available for public use, data harmonization efforts are common in order to create harmonized variables in the combined data set.

This study is motivated by the different cognitive measurements in the Neuropsychological Test Battery of the National Alzheimer's Coordinating Center (NACC) Uniform Data Set (UDS). The NACC UDS has the largest number of participants for studying the progression to mild cognitive impairment and dementia in the United States \citep{Weintraub2009TheBattery, Besser2018VersionSet}.  The batteries used for neuropsychological assessment of an individual change over time.  In 2015, the third version of the Uniform Data Set developed a non-proprietary neuropsychological battery which may be used to standardise results across individuals going forward, while prior to that time, a different set of proprietary neuropsychological battery was given.  Additionally for comparison across tests, a small group of cognitively normal individuals received both the old score battery and new battery in a crosswalk study \citep{Monsell2016ResultsStudy} which will serve as a validation sample of various harmonization methods.

In general, data harmonization requires two main steps. The first of which is to determine whether two variables $Y \in \{0,1,2,\dots, N_Y\}$ and $Z \in \{0,1,2,\dots, N_Z\}$ measure the same trait.  That is, whether their distributions can be commonly represented by a single latent variable $\gamma$.  This may either require expert knowledge or studies of individuals completing both tests.  In our case, the change in battery and their correspondence in measuring cognitive domains were guided by the Clinical Task Force (CTF), a group formed by the National Institutes of Aging to develop standardized methods for collecting longitudinal data that would encourage and support collaboration across the Alzheimer's Disease Research Centers. 
The second step is to propose a joint model of the scores, frequently using a latent variable representation \citep{Griffith2013HarmonizationMeta-Analysis, vandenHeuvel2020LatentMemory}.

These methods, however, tend to mostly rely on parametric assumptions \citep{Griffith2013HarmonizationMeta-Analysis}.   Additionally, rather than focusing on inference associated how the difficiculty of tests vary with covariates, we are interested in predicting the outcome of a test $Z$ given an individual's score on test $Y$.  

In this paper, we consider a non-parametric approach to latent variable modelling for the purpose of data harmonization. We connect this approach to mixing distribution estimation, for which there is a rich history.   One of the earliest versions of the problem arose from educational testing as estimating a population of parameters for individuals taking a test where scores for each individual are assumed follow a binomial distribution with parameters $N, p_i$ \citep{Lord1965AApplications, Lord1969EstimatingProblem}, where $p_i$ follows a certain distribution. The binomial case is a particularly well studied application of mixing distribution estimation \citep{Lindsay1983TheFamily, Wood1999BinomialDistribution, Tian2017LearningParameters, Vinayak2019MaximumParameters}. Estimation of a mixing distribution in general has numerous other applications such as in positron emission tomography \citep{Vardi1985ATomography, Silverman1990ATomographyb}, portfolio optimization \citep{Cover1984AnReturn} and ecology \citep{Bell2000EnvironmentalSedges}. In cases where $Y$ is continuous, de-convolution kernels are a common approach \citep{Leonard1990DeconvolutingEstimators, Basulto-Elias2021BivariateData} as this becomes a more standard measurement error problem. Recently there has also been an operator theoretic study of mixing distribution estimation \citep{VanDermeulen2019AnModels}.

The nonparametric EM algorithm \citep{Laird1978NonparametricDistribution} is a classical approach to compute the nonparametric maximum likelihood estimator (NPMLE) \citep{Kiefer1956ConsistencyParameters}.  An assumption was made by Laird initially that the maximum likelihood estimator should be a discrete distribution.  This problem was formalized by \cite{Lindsay1983TheTheory} who proved that under some very mild conditions there always exists a discrete maximum likelihood estimator of the mixing distribution, even if the true mixing distribution is continuous. This phenomena is partly reconciled by the non-identifiability of the maximum likelihood estimate for the mixing distribution.  However, when the goal is data harmonization, a discrete estimated mixture distribution may be undesirable, since in data harmonization, it is often assumed that a continuous latent map exists between latent variables and a discrete latent distribution would violate this assumption \citep{vandenHeuvel2020LatentMemory}. We will address this problem by a regularization approach.

The remainder of the paper will outline our statistical approach to this data harmonization problem.  First we introduce a graphical representation for our problem. We focus on non-parametric estimation of the mixing distributions, as well as introduce a regularized estimator.  We further focus on the aspects of picking the conditional distribution (measurement error) assumption required for modelling, as well as testing the feasibility of the assumption.

\section{Methods}
Consider a set of \emph{i.i.d.} random variables $(X_i, Y_i, Z_i, A_i)$, where $X_i \in \mathbb{R}^d$ is a set of covariates and $Y_i \in \{0, \dots, N_Y\}$ is the old score and $Z_i \in \{0, \dots, N_Z\}$ is the new scores.  We drop the subscript on $N$ if it is clear from the context. We focus on discrete test scores as the neuropsychological tests considered in the application are all discrete. Lastly, $A_i \in \{0,1\}$ is an indicator determining which score is observed.  In this setting only one score is observed for each individual, denote $A_i = 0$ if score $Y_i$ is observed and  $A_i = 1$ if score $Z_i$ is observed.  
Our motivation comes from heterogeneous measurements of a common trait. In our application, we think of this as two test scores measuring the same cognitive area, such as language or memory.  We next introduce our fundamental assumption of data harmonization. 

\begin{definition}[Harmonizable Variables]
    Two discrete variables $Y$ and $Z$ are \textbf{harmonizable} if (a) they are independent of all other variables, conditional on a continuous latent variables $\gamma, \zeta \in [0,1]$ respectively, and (b) conditioned on covariates $X$, there is an invertible and increasing mapping between the two latent variables $\phi_{\gamma \to \zeta}$ such that for all $X$
    \begin{equation}
        \zeta = \phi_{\gamma \to \zeta}\left( \gamma|X\right).
    \end{equation}
\end{definition}

\begin{theorem} \label{thm:cdf}
    If $Y$ and $Z$ are harmonizable variables with latent variables $\gamma, \zeta$, 
    the mappings $\phi_{\gamma \to \zeta}(\cdot|x)$ and $\phi_{ \zeta \to \gamma}(\cdot|x)$ uniquely exist and can be defined below: 
    \begin{align}
        \phi_{\gamma \to \zeta}(q|x) &= G^{-1}\left(  F(q|x) |x\right) \\
        \phi_{\zeta \to \gamma}(q|x) &=
        F^{-1}\left(  G(q|x) |x\right) =
        \phi^{-1}_{\gamma \to \zeta}(q|x)
    \end{align}
    where $F(\cdot|x)$ and $G(\cdot|x)$ are the conditional CDF's of $\gamma$ and $\zeta$ respectively
\end{theorem}

\begin{proof}
    We assume $\gamma$ and $\zeta$ must have continuous conditional CDF's then: 
    \al{
         G(q|x) &= P(\zeta \leq q|x) \\
         &= P(\phi_{\gamma \to \zeta}(\gamma|x) \leq q|x) \\
         &= P( \gamma \leq \phi^{-1}_{\gamma \to \zeta}(q|x)|x) \\
         &= F(\phi^{-1}_{\gamma \to \zeta}(q|x)|x) \\
         \implies \phi^{-1}_{\gamma \to \zeta}(q|x) &= F^{-1}( G(q|x)|x) 
         \text{ and } \phi_{\gamma \to \zeta}(q|x) = G^{-1}( F(q|x)|x)
    }
\end{proof}

If two random variables are harmonizable, they must be measuring the same trait.  This assumption is similar to that of rank preserving models though we do not assume an constant additive effect \citep{Robins1991CorrectingModels, White1997ImpactTrial, White1999Randomization-basedTrial}. 
A simple generative model for a pair of harmonizable variables 
is the following hierarchical model:
\begin{equation}
\begin{aligned}
    \Omega_i &\sim U(0,1),  &X_i& \sim P_X \\
    \gamma_i &= F^{-1}(\Omega_i|X_i), 
    &\zeta_i& = G^{-1}(\Omega_i|X_i) \\
    Y_i|\gamma_i &\sim P_{A_Y}(\cdot|\gamma_i), &Z_i&|\zeta_i \sim P_{A_Z}(\cdot|\zeta_i)\\
    A&\indep (\Omega, Y, Z)|X
\end{aligned}
\label{eq:generative}
\end{equation}
where $U(0,1)$ represents a uniform random variable, $P_X$ represents the distribution of the covariates, $F^{-1}(\cdot|x),G^{-1}(\cdot|x)$ represent the inverse of the conditional cdfs of $\gamma$ and $\zeta$ respectively and $P_{A_{Y}}$ (or $P_{A_{Z}}$) represents the corresponding distribution of the observed score, given the latent trait (i.e. the measurement error). 
We summarize this hierarchical model using Figure~\ref{fig:gm}.  This is not a conventional graphical model, but we use the graph to summarize the generative assumption of our model. The square variables denote that these variables are deterministic functions of their parent variables.

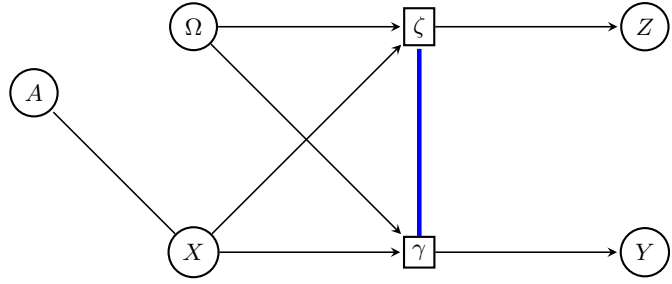
\begin{figure}[tb!]
  \begin{center}
        \begin{tikzpicture}[
                > = stealth, 
                shorten > = 1pt, 
                auto,
                node distance = 3cm, 
                scale = 1,
                transform shape,
                semithick 
            ]
        
            \tikzstyle{every state}=[
                draw = black,
                thick,
                fill = white,
                minimum size = 4mm
            ]
        
            \node[state] (x) {$X$};
            \node[state] (o) [above of= x] {$\Omega$};
            \node[state] (a) [above left of= x] {$A$};
            \node[state] (y) [rectangle,right of=x] {$\gamma$};
            \node[state] (z) [rectangle, right of=o] {$\zeta$};
            \node[state] (y2) [right of=y] {$Y$};
            \node[state] (z2) [right of=z] {$Z$}; 
            
            \path[->] (x) edge node {} (y);
            \path[->] (x) edge node {} (z);
            \path[-]  (x) edge node {} (a);
            \path[->] (o) edge node {} (y);
            \path[->] (o) edge node {} (z);
            \path[->] (y) edge node {} (y2);
            \path[->] (z) edge node {} (z2);
            \path[-, ultra thick, draw=blue] (y) edge node [midway, label=above:$ $] {} (z);
            \path[-, ultra thick, draw=blue] (y) edge node [midway, label=below:$  $] {} (z);
       \end{tikzpicture}%
    \caption{Score Generation Model. Deterministic relationships between latent variables shown in blue.  Deterministic variables dependent on upstream variables shown with squares.}
    \label{fig:gm}
  \end{center}
\end{figure}
The generative model in equation \eqref{eq:generative} can be interpreted as follows.
$\Omega$ is the relative quantile of an individual's cognitive ability in the population. 
Thus, $\Omega$ together with the covariate $X$ determines the latent traits  $\gamma, \zeta$ of the two tests.
The observed scores $Y$ and $Z$ are the measured version of the latent traits, so they are completely determined by $\gamma$ and $\zeta$, respectively. 
The conditional independence $A\indep (\Omega, Y,Z)|X$ is motivated by the fact that the score conversion in the NACC data occurs at 2015, which is not relevant to any individual's ability (but it may be relevant to the year of visit/age, which is part of the covariate $X$),


Although we only observe either $(X_i,Y_i, A_i = 0)$ or $(X_i,Z_i, A_i = 1)$, 
the generative model in equation \eqref{eq:generative} implies that 
the observed data is determined by 
$p(y|\gamma), p(z|\zeta), p(\gamma|x), p(\zeta|x)$, $p(a|x)$ and $p(x)$. 
For the purpose of data harmonization, we do not need to model $p(a|x)$ or $p(x)$,
so we will focus on modeling the first four distributions.
The first two distributions $p(y|\gamma), p(z|\zeta)$
are the measurement \textbf{A}ssumptions, so we denote them as $p_A(y|\gamma), p_A(z|\zeta)$.
The two distributions $p(\gamma|x), p(\zeta|x)$ are the latent trait \textbf{M}odels, so we denote them as $p_M(\gamma|x), p_M(\zeta|x)$.
In what follows, we describe how we model $p_A$ and $p_M$.
To simplify the notations, we focus on $p_A(y|\gamma)$ and $p_M(\gamma|x)$;
the case of $p_A(z|\zeta)$ and $p_M(\zeta|x)$ can be modeled in a similar manner. 


{\bf Remark.} We can immediately draw a comparison from this representation to the Skorohod notation of a random variable which are introduced in quantile regression models (not to be confused with the Skorohod representation theorem) \citep{Chernozhukov2006InstrumentalModels, Kato2012EstimationRegression}.


{\bf Remark.}
We contrast this with the approach of \citep{vandenHeuvel2020LatentMemory}, in which they assume that the distribution of $Y|\gamma$ may not be independent of $X$, however the mapping between latent variables $\phi_{\gamma \to \zeta}$ is independent of $X$.  Our measurement assumption is similar to \citep{Meredith1993MeasurementInvariance} in which $p(y|\gamma,x) = p(y|\gamma)$ and similarly for $Z$.


\subsection{Measurement Assumption Model}

A conventional approach for modeling $p_A$
is the binomial model \citep{Lindsay1983TheFamily, Wood1999BinomialDistribution,Tian2017LearningParameters, Vinayak2019MaximumParameters}, for measurements that are sum scores of a fixed number of binary question items.  However, this may be a restrictive assumption as it inherently depends on equal difficulties for each test question with independent responses.  If we had access to all the question items, it would be possible to learn a conditional model with the difficulties for each question using a Rasch model \citep{Rasch1966AnAnalysis, Lindsay1991SemiparametricAnalysis} or a more elaborate IRT model.  In the NACC UDS application, the individual binary question items are not available.  The measurement assumption model can be selected by a domain expert, but additionally we can also test for its feasibility as we present in Section ~\ref{sec:feasibility_test} as well as present a method for selecting this model from the data in Section \ref{sec:error_select}. 

In this paper, we propose a flexible method to construct a measurement assumption
by measurement
kernel (functions).
Specifically, we model the measurement assumption as
\begin{equation}
    p_A(y|\gamma) \propto K\left(\frac{y - N\gamma}{h}\right),
\end{equation}
where $K$ is the measurement kernel and $h>0$ is the smoothing bandwidth.
This model enables great flexibility in terms of the shape of the error distribution, controlled by $K$ and the spread of the distribution, controlled by $h$. We also denote $Y|\gamma \sim MKM(K, h, \gamma)$ if the conditional distribution is defined by the model above.  The kernel function defines the relative probability of the conditional distribution $Y|\gamma$ with assigning higher probability for $y$ near $N\gamma$. $K$ defines the decay of the relative probability as $y$ is further from $N\gamma$
\subsection{Latent Trait Models}

\subsubsection{Non-Parametric Model} \label{sec:nonpara_model}

In many applications, parametric models for a latent model may not be particularly realistic.  As a result, we would like to relax this assumption and consider a non-parametric method.  We propose a model based on the non-parametric EM algorithm \citep{Laird1978NonparametricDistribution}. 
To simplify the derivation, we omit the covariate $X$.
We illustrate how we can incorporate covariates in Section \ref{sec:include_covariates}.  Estimating the latent model $p_M(\gamma)$ now clearly becomes a mixing distribution problem with log-likelihood: 
\begin{equation} \label{eq:mixing_likelihood}
    \ell_n[P_M] = \frac{1}{n}\sum_{i = 1}^n \log\left( \int p_A(y_i|\gamma) p_M(\gamma) d \gamma \right) \quad P_M \in \mathcal{P}_{[0,1]}
\end{equation}
where we differentiate the parametric likelihood $\ell_n(\theta)$ from the functional likelihood $\ell_n[P_M]$, with $\mathcal{P}_{[0,1]}$ denotes the space of all probability measures over $[0,1]$.  We also highlight that $P_M$ refers to the measure which has a corresponding generalized density function $p_M$ for which we may allow point masses. 
This problem involves solving the NPMLE \citep{Kiefer1956ConsistencyParameters} and was the original inspiration for the nonparametric EM algorithm (NPEM, also known as the functional EM) \citep{Laird1978NonparametricDistribution}.

With a basic rearrangement, we can rewrite equation  (\ref{eq:mixing_likelihood}), the mixture likelihood, to the following:
\begin{equation}
    \ell_n[P_M] = \sum_{y = 0}^N \hat p(y)\log\left( \int p_A(y_i|\gamma) p_M(\gamma) d \gamma \right)
\end{equation}
where $\hat p(y)$ is the empirical distribution of $y$
\begin{equation}
    \hat p(y) = \frac{1}{n}\sum_{i = 1}^n I(Y_i = y).
\end{equation}

This leads us to the standard nonparametric EM algorithm \citep{Laird1978NonparametricDistribution}
\begin{equation}\label{eq:npem}
    p_M^{(j + 1)}(\gamma) = \E_{\hat p(y)}[p_M^{(j)}(\gamma|y)]= \sum_{y = 0}^N\frac{p_A(y|\gamma)\hat p(y)p^{(j)}(\gamma)}{p_{MA}^{(j)}(y)}
\end{equation}
where $p_{MA}^{(j)}(y) = \int p_A(y|\gamma)p_M^{(j)}(\gamma)d\gamma$.

{\bf Remark.}
\citep{Lindsay1983TheTheory} studied the nonparametric EM in depth and proved that if $Y$ has a support size of $N+1$, then there will always exists a discrete NPMLE with at most $N + 1$ masses, even under a true model where $p_0(\gamma)$ is continuous.  This phenomena is explained by the non-identifiability of this model in general, which we will address further in section ~\ref{sec:reg_nplt}.  Approaches to ensure smoothness have also been developed for the latent density estimation such as early stopping of an EM algorithm \citep{Chae2017FastDensity}, imposing a roughness penalty on the density likelihood \citep{Liu2009AMaximizationb} or using a kernel density smoother between iterations of an EM algorithm \citep{Silverman1990ATomography}.  The early stopping is not ideal as we are intentionally not fully maximizing the likelihood.  The kernel density method of \cite{Silverman1990ATomography} also has problems of how to choose the smoother bandwith, as well as losing the non-decreasing property of the EM algorithm. The roughness penalized method of \cite{Liu2009AMaximizationb} is challenging to compute as they introduce an EM algorithm which requires solving an ODE at each iteration.  

Under the conditions of \cite{Chung2015ConvergenceDistributions}, the update in equation \eqref{eq:npem} leads to an improved estimator in the sense that the implied marginal distribution of $Y$ is closer to the empirical distribution in relative entropy. We next introduce a theorem that under a list of sharpness conditions on $p_A$ which are listed in supplement A, we will have linear convergence.

\begin{theorem} \label{thm:linear_convergence}
    Under the assumptions listed in the supplementary material part A, the NPEM algorithm will converge linearly (in the 1-norm) with rate $\rho \in [0,1)$ 
    \begin{equation*}
        \norm{p_{MA}^{(t + 1)}(y) - \hat p(y)}_1 \leq \rho \norm{p_{MA}^{(t)}(y) - \hat p(y)}_1.
    \end{equation*}
   
\end{theorem}

Though the assumptions required for linear convergence are rather strict, they reflect a broader trend seen empirically where a sharper measurement assumption model $p_A$ tends to converge faster with the nonparametric EM algorithm. In general, this algorithm must be approximated since we  do not have an exact evaluation of the integral to compute  $p^{(j + 1)}_M(\gamma)$. We illustrate an approximation of this method in the supplementary material part A. 

The nonparametric likelihood in equation \eqref{eq:npem} suffers from
a non-identifiability problem as the likelihood functional $\ell_n[P_M]$ is only concave in $P_M$,  not strictly concave.
Namely, there might be multiple distributions which maximize $\ell_n[P_M]$. 
Due to the non-identifiability of a maximizer, the convergence properties of the nonparametric EM algorithm are far less studied than in the parametric case.
As far as we are aware,
\cite{Chung2015ConvergenceDistributions} is
the only work that
proved the convergence of nonparametric EM in terms of likelihood values
but it did not derive the rate as Theorem~\ref{thm:linear_convergence}.

Though the mixing distribution $P_M$ is not identifiable, we still may find that the imposed distribution on the marginal $p_{MA}$ will converge linearly to $\hat p$, the empirical distribution under the regularity conditions outlined in the supplement.

\subsubsection{Regularized Nonparametric Latent Trait} \label{sec:reg_nplt}

To deal with the non-identifiability issue, we introduce a regularized NPEM algorithm leading to a unique maximizer.  Let $P_U$ be the measure of a uniform random variable over [0,1]. We introduce a regularization term via the KL-divergence from the uniform distribution, i.e., $\mathcal{D}(P_U||P_M)$. 

We define the regularized likelihood as 
\al{
    \ell_{n,\mu}[P_M] &= \ell_{n}[P_M] - \mu \mathcal{D}(P_U||P_M). 
}

Due to the strict convexity of $\mathcal{D}(P_U||P_M)$ in each parameter, the regularized objective has a unique maximizer whenever $\mu > 0$.  One reason of regularizing to a uniform density is due to its reasonable choice as an uninformative latent distribution much like the choice of regularizing coefficients to $0$ in the LASSO \citep{Tibshirani1996RegressionLasso}. Additionally, since a discrete latent distribution $p_M$ will have infinite KL-divergence from the uniform $\mathcal{D}(P_U||P_M)$, regularizing toward a uniform distribution will smooth the estimate depending on the parameter $\mu$. Moreover, this choice of regularization penalty leads to computational convenience. We outline two methods which can be used to compute the regularized nonparametric MLE (RNPMLE). 

{\bf Optimization Algorithm 1 : Regularized NPEM algorithm.}
The NPEM algorithm only requires a minor change for the likelihood which includes a regularization term, and this NPEM algorithm amounts to estimating a mixture between the standard EM step and with a uniform distribution. 

We derive the form of the update much like the original EM algorithm:
\al{
        \ell_{n,\mu}[P_M] &= \sum_{y = 0}^N \hat p(y)\log\left( \int p_A(y|\gamma)p_M(\gamma)d\gamma\right) - \mu \mathcal{D}(P_U||P_M) \\
        &=  \sum_{y = 0}^N \hat p(y)\log\left( p_{MA}(y)\right) - \mu \mathcal{D}(P_U||P_M) \\
        &=  \sum_{y = 0}^N \hat p(y)\left( \log\left( p_{MA}(y,
        \gamma)\right)  - \log\left( p_{MA}(\gamma|y)\right)\right) - \mu \mathcal{D}(P_U||P_M) \\
        &= \sum_{y = 0}^N \hat p(y)\int(p_{MA}^{(j)}(\gamma|y)\left( \log\left( p_{MA}(y,
        \gamma)\right)  - \log\left( p_{MA}(\gamma|y)\right)\right)d\gamma - \mu \mathcal{D}(P_U||P_M) \\
        &= Q[P_M||P^{(j)}_M] + H[P_M||P^{(j)}_M] - \mu \mathcal{D}(P_U||P_M),
}
where 
\al{
    p_{MA}(y,\gamma) &= p_{A}(y|\gamma)p_M(\gamma),\quad  p_{MA}(\gamma|y) = \frac{p_{MA}(y,\gamma)}{p_{MA}(y)} \\
    Q[P_M||P^{(j)}_M] &=\sum_{y = 0}^N \hat p(y)\int p_{MA}^{(j)}(\gamma|y)\log\left( p_{MA}(y,\gamma)\right)  d\gamma,  \\
    H[P_M||P^{(j)}_M] &= -\sum_{y = 0}^N \hat p(y)\int p_{MA}^{(j)}(\gamma|y)\log\left( p_{MA}(\gamma|y)\right)  d\gamma
}
By Gibbs' Inequality $H[P_M||P^{(j)}_M] \geq H[P^{(j)}_M||P^{(j)}_M]$ with equality only if $P_M = P^{(j)}_M$. 
Thus, 
$$
 \ell_{n,\mu}[P_M] \geq Q[P_M||P^{(j)}_M] - \mu \mathcal{D}(P_U||P_M)
$$
and the regularized nonparametric EM
is 
\al{
    p^{(j+1)}_M = \argmax_{p_M \in \mathcal{P}_{[0,1]}} Q[P_M||P^{(j)}_M] - \mu \mathcal{D}(P_U||P_M),
}
which will guarantee that the likelihood increases with each iteration.  
In addition, observe that 
\al{
     Q[P_M||P^{(j)}_M] - \mu \mathcal{D}(P_U||P_M) &= \int \sum_{y = 0}^N \hat p(y)\frac{p_A(y|\gamma)p^{(j)}_M(\gamma)}{p^{(j)}_{MA}(y)} \log(p_A(y|\gamma)p_M(\gamma)) d\gamma +  \int \log(p_A(y|\gamma)) d\gamma \\
     &= (1 + \mu)\int \left( \sum_{y = 0}^N \hat p(y)\frac{p_A(y|\gamma)p^{(j)}_M(\gamma)}{(1 + \mu)p^{(j)}_{MA}(y)} + \frac{\mu}{1 + \mu}\right)\log(p_M(\gamma)) d\gamma + C
}
and once again by Gibbs' inequality, we can uniquely maximize $Q[P_M||P^{(j)}_M] - \mu \mathcal{D}(P_U||P_M)$ by letting
\begin{equation}
    p_{M}^{(j + 1)}(\gamma) = \frac{1}{1 + \mu}\E_{\hat p(y)}[p_{M,}^{(j)}(\gamma|y)] + \frac{\mu}{1 + \mu}= \sum_{y = 0}^N\frac{p_A(y|\gamma)\hat p(y)p_{M,}^{(j)}(\gamma)}{(1 + \mu)p_{MA, \mu}^{(j)}(y)} + \frac{\mu}{1 + \mu}.
    \label{eq:npem2}
\end{equation}
When $\mu = 0$ this simply reduces to the nonparametric EM algorithm in equation \eqref{eq:npem}. 

Next we present a theorem for the convergence of the EM algorithm on the regularized likelihood. 

\begin{theorem} \label{thm:regularized_EM}
    Denote the unique global solution $P^*_{M,\mu} = \argmax_{P_M \in \mathcal{P}_{[0,1]}} \ell_{n,\mu}[P_M]$.
    
    Then consider a sequence of latent trait distributions $\{P^{(t)}_{M,\mu}\}_{t = 0}^{\infty}$ generated by the EM algorithm for the regularized likelihood. If $\mu > 0$ and $p_A(y|\gamma)$ is continuous in $\gamma$ for each $y$, then 
    
    $$ \ell_{n,\mu}[P^{(t)}_{M,\mu}] \stackrel{t \to \infty}{\longrightarrow} \ell_{n,\mu}[P^{*}_{M,\mu}].$$ 
    
    and
    
    $$ P^{(t)}_{M,\mu} \stackrel{t \to \infty}{\longrightarrow_w} P^{*}_{M,\mu}$$
    where $\longrightarrow_w$ denotes weak convergence of measures. 
\end{theorem}

The proofs are in the supplementary material and rely on techniques for proving the the convergence of the unregularized EM algorithm \citep{Chung2015ConvergenceDistributions}, as well as convex optimization in infinite dimensional vector spaces \citep{Kosmol2011OptimizationSpaces}.  Once again, in practice, we must use an approximation to the regularized EM algorithm.  

With the inclusion of a small perturbation in the likelihood we gain uniqueness of a maximizer, and weak convergence of an EM algorithm no matter the measurement assumption model $p_A(y|\gamma)$ whereas in Theorem \ref{thm:linear_convergence} we have shown that when $p_A(y|\gamma)$ is a sharp model, we will have fast convergence of the EM algorithm. The assumptions required in Theorem \ref{thm:linear_convergence} are far more strict, but those required for the convergence results in Theorem \ref{thm:regularized_EM} are very minimal. However, this algorithm may converge very slowly in general, and thus a fast approximation is also desirable.

{\bf Optimization Algorithm 2 : Fast computation via geometric programming with binning.}

Although the nonparametric EM algorithms in equations \eqref{eq:npem} and \eqref{eq:npem2} have nice properties, 
we are optimizing functions, which are infinite dimensional objects.
Thus, the nonparametric EM may still be very slow in practice.
To reduce the computational cost,
we consider a binning approximation of 
the latent trait distribution.  Recall that the regularized likelihood is
\begin{equation*}
    \ell_{n,\mu}[P_M] = \sum_{y = 0}^N \hat p(y)\log \left( \int p_A(y|\gamma)p_M(\gamma)d\gamma\right) + \mu \int_{0}^1 \log(p_M(\gamma))d\gamma 
\end{equation*}
We consider a binned version of the latent density:
\begin{align*}
    p_M(\gamma) &=  R\theta_{r} \quad   \gamma \in \left[\frac{r - 1}{R}, \frac{r}{R}\right) \quad r \in \{1,2,\dots, R\}.
\end{align*}
Namely, $\theta_{r}$ is the weight of the mixing density between the points $[\frac{r - 1}{R}, \frac{r}{R})$.

Let $A \in \mathbb{R}^{(N + 1) \times R}$ be a matrix where $A_{ij} = \int_{\frac{j - 1}{R}}^{\frac{j}{R}}p_A(i|\gamma)d\gamma$ and $A_{y \cdot}\in \R^{R}$ be the row vector of $A$ at position $y$.  Our discrete approximation to the maximization of $\ell_{n,\mu}[P_M]$ can now be written as 
\begin{equation*}
\begin{aligned}
    & \underset{\theta}{\text{sup}}
    & & \sum_{y = 0}^N \hat p(y)\log(A^T_{y \cdot}\theta) + \mu \frac{1}{R}\sum_{r = 1}^R\log(R\theta_r) \\
    & \text{subject to}
    & & 1^T\theta = 1 \\
    &&& \theta \succeq 0.
\end{aligned}
\end{equation*}
The maximization problem can be solved by standard convex solvers since this problem is equivalent to a geometric program.  In our numerical analysis in Sections \ref{sec:Simulations}, \ref{sec:application}, we implement this with the newly developed CVXR \citep{Fu2020CVXR:Optimization} with the MOSEK convex solver \citep{Andersen2000TheAlgorithm}. 

In the supplemental material section \ref{supp:sec:speed_test} , we illustrate that this geometric program is much faster than the Regularized NPEM algorithm. 

\subsubsection{Incorporating covariates} \label{sec:include_covariates}
To incorporate covariates, we replace the estimate of $\hat p(y)$ with a conditional estimate $\hat p(y|x)$.  
With this, we modify the log-likelihood $\ell_{n,\mu}$ to be
the conditional log-likelihood:
\[
    \ell_{n,\mu}[p_M;x] =  \sum_{y = 0}^N \hat p(y|x)\log \left( \int p_A(y|\gamma)p_M(\gamma|x)d\gamma\right) + \mu \int_{0}^1 \log(p_M(\gamma|x))d\gamma.
\]

Note that $p_M(\gamma|x)$ now depends on $x$ as well. 
For the the discretized latent approximation we replace $\theta$ with $\theta(x)$. 

\subsection{Score Conversion} 

By applying the above procedure to both $(X,Y, A=0)$ and $(X,Z, A=1)$,
we obtain estimators of $p_M(\gamma|x)$ and $p_{M}(\zeta|x)$.
By Theorem \ref{thm:cdf}, these quantities
can be used to create the conversion $\phi_{\gamma \to \zeta}(\cdot|x)$ and $\phi_{\zeta \to \gamma}(\cdot|x)$, which further leads to a conversion between $Y$ and $Z$ as described in Algorithm ~\ref{alg:Conversion_sample}.  This amounts to constructing an estimate for $p(\gamma|y,x)$ based on the estimated mixing distributions and the measurement assumptions.  Treating the estimated mixing distributions as priors, our conversion of scores is essentially an empirical Bayes method \citep{Laird1982EmpiricalPrior}.


\begin{algorithm}[H] \label{alg:Conversion_sample}
    Set a fixed sampling number $J$ \\
    \For{ $i \in \{1,\dots, n\}$ }{
        \For{ $j \in \{1,\dots, J\}$ }{
              a. Sample $\widetilde \gamma_{ij}$ from $p_{\widehat M(\mu)A_Y}(\gamma|X_{i},Y_{i})$ \\
              b. Convert $\widetilde \zeta_{ij} = \phi_{\gamma \to \zeta}(\widetilde \gamma_{ij}|x_i)$ \\
              c. Sample $\widetilde Z_{ij}$ from $p_{A_Z}(Z|\widetilde \zeta_{ij})$
        }
    
       Denote $\mathcal{Z}_i = \{\widetilde Z_{ij}\}_{j = 1}^J$ 
     }
    
 \caption{Conversion Sampling } \
\end{algorithm}

In practice, we must approximate $\phi_{\gamma \to \zeta}(\cdot|x)$ as we do not have analytic forms for $F(\cdot|x)$ and $G(\cdot|x)$.   A simple method is to let $\hat \phi^{L}_{\gamma \to \zeta}(\cdot|x)$ be a piece-wise linear function constructed by the linearized versions of $F(\cdot|x)$ and $G(\cdot|x)$.  This is a simple to compute method which can be made arbitrarily precise depending on the number of quantiles chosen.  In the case of a binned latent distribution, the estimates of $\hat F_{M(\mu)}(\cdot|x)$ and $\hat G_{M(\mu)}(\cdot|x)$ are already piece-wise linear as the corresponding densities are piece-wise constant. 

We illustrate these details further in the supplementary material section \ref{supp:sec:computational_details}.

\section{Model diagnostics and model selection with multiple observations}

When we only have one observation per person, 
it is difficult to check if our generative
model and the measurement kernel are compatible with the data.
However, the NACC UDS is longitudinal and we have multiple observations per individual. 
In this section, we will discuss how
we may use two consecutive observations of the same
individual to test the feasibility of a measurement kernel
and choose the model and tuning parameter in our analysis.
Note that we restrict ourselves to only two consecutive observations per individual because these observations are yearly information of the same individual. It is reasonable to assume that the cognitive ability of the same person is approximately constant in this short window.

\subsection{Feasibility test of the measurement assumption model} \label{sec:feasibility_test}

In this section, we propose two simple tests to examine if the measurement assumption is reasonable. We use this test as a model diagnostic procedure in practice. 
The goal is to examine if the measurement assumption $p_A(y|\gamma)$ is in agreement with the observed data or not.

{\bf First-order feasibility test.}
For a given model $p_A(y|\gamma)$, it implies a marginal distribution $p_{MA}(y) = \int p_A(y|\gamma) p_M(\gamma)d\gamma$. 
The marginal distribution $p_{MA}(y)$ can be compared to
the empirical distribution $\hat p(y)$. 
In practice, we can use the KL divergence between $\hat p$ and $\hat p_{MA}$ 
to investigate if such a measurement assumption $p_A(y|\gamma)$ is feasible or not. 

Geometrically, the first-order feasibility test
can be understood as follows. 
Let $W = (W_0,\cdots, W_N)\in\R^{N+1}$
be a probability vector, i.e., $\sum_{j=0}^N W_j = 1$ and $W_j\geq 0$. 
Clearly, $W\in\mathcal{S}_{N+1}$,
where $\mathcal{S}_{N+1}$ is the $(N+1)$-simplex.
The empirical distribution $\hat p(y)$
is a point $\hat p = (\hat p(0),\cdots, \hat p(N)) \in \mathcal{S}_{N+1}$.
At a given $\gamma$, the measurement assumption $p_A(y|\gamma)$ is also
an element $p_A(\cdot|\gamma) = (p_A(0|\gamma),\cdots, p_A(N|\gamma))\in\mathcal{S}_{N+1}$.
By the same construction,
the implied marginal distribution $p_{MA}\in\mathcal{S}_{N+1}$.
While different latent distribution $p_M(\gamma)$ leads to a different marginal $p_{MA}$, 
it  is easy to see that the collection
of all possible marginal distribution from a measurement assumption $p_A$
is ${\sf conv}(\Gamma)$, the convex combination of
$$
\Gamma = (p_A(\vect{y}|\gamma): \gamma\in[0,1])\subset \mathcal{S}_{N+1}.
$$
See Figure~\ref{fig:feasibility1} for an example.  

We test for population feasibility, i.e. whether $p_0 \in \text{conv}(\Gamma)$ by checking the closeness of $\text{conv}(\Gamma)$ to $\hat p$.  This can be done by fitting the NPMLE.  Though the estimated mixing distribution $\hat P_M$ may not be unique, the marginal implied distribution $\hat p_{MA}$ will be unique.  This is due to the fact that maximum likelihood estimation is equivalent to finding a particular mixing distribution $\hat P_M$ which minimizes the relative entropy distance between $\hat p$ and $p_{MA} \in \text{conv}(\Gamma)$ while $\mathcal{D}(\hat p||\cdot)$ is strictly convex.  See figure \ref{fig:feasibility1} for an illustration of the path $\Gamma$ and the convex hull.


{\bf Second-order feasibility test.}
When we have two observations of the same individual (which occurs 
in the NACC data since it is a longitudinal database), 
we can generalize the above procedure.  We assume that an individual has a constant trait $\gamma$ between measurements. Our model then describes a restriction on the distribution of the pairs of measurements. Similarly let $W_{(2)} = (W_{0,0}, W_{0,1}, \dots, W_{N,N})$ be a probability vector $W_{(2)} \in \mathcal{S}_{(N + 1)^2}$.  We define $\hat p_2 = (\hat p(0,0), \hat p(0,1), \dots, \hat p(N,N)) \in \mathcal{S}_{(N + 1)^2}$ as the empirical distribution of the pairs of observations. In the case where two observations are generated from an individual with a single $\gamma$, we can fit a latent distribution on two observations with the following mixture likelihood: 
\al{
    \ell_{2,n}[P_M] &= \frac{1}{n} \sum_{i = 1}^n \log\left( \int p_A(y_{i1}|\gamma)p_A(y_{i2}|\gamma) dP_M(\gamma)\right) \quad P_M \in \mathcal{P}_{[0,1]}
}
We obtain an analogous implied marginal distribution on the distribution of pairs $(Y_{i1}, Y_{i2})$, $\hat p_{2,MA}$ from a fitter model where $\hat P_{2,M} \in \argmax \ell_{2,n}[P_M]$
\al{
    \hat p_{2,MA}(y_1,y_2) &= \int p_A(y_{1}|\gamma)p_A(y_{2}|\gamma) d\hat P_{2,M}(\gamma) d\gamma
}
As before, the collection of all bivariate distributions generated by this model with a fixed $\gamma$ can be expressed as $\text{conv}(\Gamma_2)$ where

\al{
    \Gamma_2 &= (p_A(\vect{y}_{1}|\gamma) \otimes p_A(\vect{y}_{2}|\gamma): \gamma\in[0,1]) \subset \mathcal{S}_{(N+1)^2}
}
and $\otimes$ is the outer product. 


\begin{figure}[htb!]
    \centering
    \includegraphics[scale = 0.42]{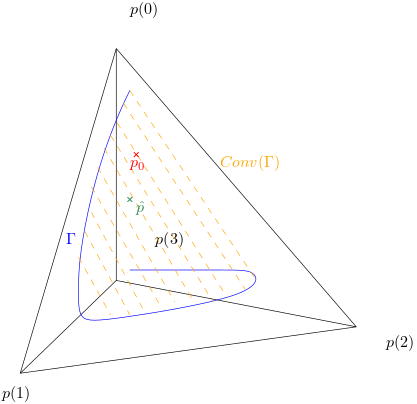}
    \caption{Example for $N = 3$ of the first order population feasibility of $p_0$ and $\hat p$. As both are in the interior of $\text{conv}(\Gamma)$}
    \label{fig:feasibility1}
\end{figure}

The reason to consider a second-order feasibility test is that the first-order test may be insensitive when the tuning parameter of the measurement model $h \to 0$.   The measurement assumption model becomes sharper and approaches a path on the boundary of $\mathcal{S}_{(N + 1)}$, meaning all population score distributions can be represented by this model $p_A(y|\gamma)$. Also as $h \to \infty$ the measurement assumption path converges to a single point $\Gamma \to p_U(\vect{y})$ where $p_U(\vect{y})$ is the uniform distribution over $N+1$ points and $\text{conv}(\Gamma)$ shrinks to that point. 



To see how the second-order feasility test can reconcile this problem, as $h \to 0$, $\Gamma_2$ approaches the edge of the simplex moving from $p(0,0) $ to $p(1,1), \dots, p(N,N)$. This is because under this model, we assume no variability given an individual's $\gamma$ score.  See figure \ref{fig:feasibility_tests} for an outline.  Once again, all the bivariate distributions expressible by a single mixing distribution can be denoted as $\text{conv}(\Gamma_2)$ and a second order population feasibility test can be interpreted as to whether $\hat p_{2,MA} (\vect{y}_1, \vect{y}_2) $ is sufficiently close to $\hat p_2(\vect{y}_1, \vect{y}_2)$. As in indicated previously, the distribution $\hat p_{2,MA}$ is the unique closest point in $\text{conv}(\Gamma_2)$ to $\hat p_2 $ in terms of the KL-divergence.   

\begin{figure}[h!]
\centering
\subfigure[First order population feasibility region]{
\includegraphics[scale = 0.42]{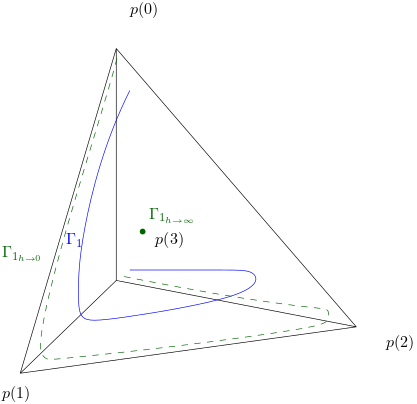}
\label{fig:subfigure1}}
\quad
\subfigure[Second order population feasibility region]{
\includegraphics[scale = 0.42]{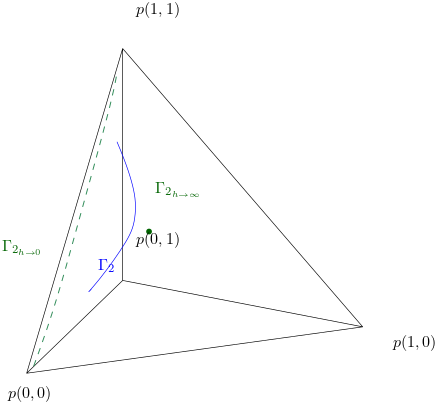}
\label{fig:subfigure2}}
\caption{Though $\text{conv}\left(\Gamma_1\right)$ tends to grow as $h$ gets large, it is clear that this places must stronger restrictions on the second order population feasibility as $\Gamma_2$ approaches a line on the boundary of $\mathcal{S}_{(N + 1)^2}$ which traces the symmetric distributions. Note that in the second order population feasibility case we can only illustrate the case of $N = 1$ and any higher order example becomes impossible to visualize.}
\label{fig:feasibility_tests}
\end{figure}

\subsubsection{Multinomial Concentration and higher-order tests}
We define the $k$-th order population feasibility tests, and denote the observed proportion vector $\hat p_k \in \mathcal{S}_{(N + 1)^k}$ generated with true probabilities $p_{0,k} \in \mathcal{S}_{(N + 1)^k}$. Denote $\mathcal{D}$ the KL-divergence or relative entropy. Under standard maximum likelihood estimation, to test whether $\hat p_k$ came from a model with parameter $p_{0,k}$ we can compute the likelihood ratio statistic
\al{
    2n\mathcal{D}(\hat p_k || p_{0,k} ) &\to_d \chi^2_{(N+1)^k - 1}
}
Since we are testing whether the set is feasible or not, instead of testing a fixed $p_{0,k}$, we rather test whether $p_{0,k} \in \text{conv}(\Gamma_k)$ (where $\Gamma_k$ is defined analogously to $\Gamma_2$). We can do this by letting replacing $p_{0,k}$ with $\hat p_{MA}$ where $ \hat P_M \in \argmax \ell_n[P_M]$.  We can use this to produce a test for the null hypothesis $H_0: p_{0,k} \in \text{Conv}(\Gamma_k)$ against the alternative $H_A: p_{0,k} \not \in \text{Conv}(\Gamma_k)$
\al{
    \mathbb{P} \left( n\mathcal{D}(\hat p_k || p_{0,k} ) > t | H_0\right) &\leq \mathbb{P} \left( n\mathcal{D}(\hat p_k || \hat p_{MA} ) > t | H_0\right)
}
An additional challenge is due to the fact the distribution of $\mathcal{D}(\hat p_k || p ) $ is asymptotic, and since the support size $\hat p$ grows exponentially in $k$ there is a particular need for a finite sample result, in particular for $k \geq 2$.  Finite sample concentration of the multinomial distribution in relative entropy is an active area of research.   We will use the recent result in \citep{Guo2021Chernoff-TypeEntropy} to compute an upper bound for $ \mathbb{P} \left( n\mathcal{D}(\hat p_k || \hat p_{MA} ) > t | H_0\right)$.  A failure of this test indicates gross misspecification of the measurement model.

We note that we use the term feasibility test rather than a hypothesis test due to the fact that for a finite $k$ we will not in general be able to discern all models $p_A(y|\gamma)$ as $n_k \to \infty$, only rule out models which do not meet the feasibility test. This test would have no power in these situations, a phenomena similar to falsification tests  \citep{Kang2013TheApproach, Wang2017OnModel, Keele2019FalsificationOperate}.

\subsection{Selection of the Measurement Assumption model using consecutive observations} \label{sec:error_select}

While the measurement  assumption model has to be assumed based on prior knowledge about the data generating process, 
when there are multiple measurements of the same individual,
we can choose it from the data. 
Here we describe a simple data-driven
procedure of choosing a measurement  assumption model
based on two consecutive observations in a longitudinal data.
Note that this idea can be generalized to multiple observations; we use
two consecutive observations because 
in the NACC data, two consecutive observations only differ by one year
and we expect individual's cognitive ability ($\gamma$) will not change much within a single year.


Let $\mathcal{A}$ be a collection of measurement assumption models. 
One can think of the element $A \in \mathcal{A}$ to be a particular choice of measurement assumption model with kernel $K$ and bandwidth $h$.  But $A$ may also include the binomial or other measurement assumption model. 
\begin{algorithm}[H]
  1. Set a fixed sampling number $J$
 \For{ $j \in \{1,\dots, J\}$ }{
    \For{ $i \in \{1,\dots, n_2\}$ }{
          a. Sample $\gamma_{ij}$ from $\hat p(\gamma|Y_{1i})$ \\
          b. Sample $\hat Y^{(A)}_{ij}$ from $p_{A}(y|\gamma_{ij})$ \\
          c. Compute $\hat E^{(A)}_{ij} = \hat Y^{(A)}_{ij} - Y_{1i}$ 
    }

 }
 2. Compute the empirical estimate of $\hat q_A (\omega) = \frac{1}{Jn_2}\sum_{ij} I(\hat E^{(A)}_{ij} = \omega)$,
 where $n_2$ is the number of individuals with second observations.\\
 \caption{Intrinsic Variability Sampling } 
 \label{alg:Intrinsic_variability_sample}
\end{algorithm}
Our model selection procedure is
is very simple. 
We use the first observation of each individual to estimate the 
latent distribution $\hat p_M(\gamma)$. 
Then for each individual, we use his/her first observation 
with the latent distribution $\hat p_M(\gamma)$
to predict value of the second observation. 
This gives us a simple prediction about the difference between
the first and the second observation, denoted as $\hat E$. 
Finally, we compare its distribution to the distribution of actual difference between the two observation
and choose the model with the best accuracy.


Specifically, 
let $Y_{1i}, Y_{2i}$ be the first and second observations of i-th individual.
Let 
\al{
    E_{i} &= Y_{1i} - Y_{2i} \in \{-N, -(N - 1), \dots, N\} 
}
be the difference between the first and the second observation. 
One can think of $E_i$ as IID from an unknown distribution $q_0$. 
Let $\hat q$ be the EDF of $E_i$. 
To examine how a measurement assumption model $A$
fits to the data,
we use Algorithm ~\ref{alg:Intrinsic_variability_sample}
that generates an estimate of the distribution $\hat q_A$
based on the model $p_A(y|\gamma)$. 
We include details on sampling from all relevant distributions in the supplementary material.
Finally, 
we choose the optimal $A$ such that
\begin{equation*}
    A^* = \argmin_{A \in \mathcal{A}} \mathcal{\tilde D}(\hat q, \hat q_A) \ ,
\end{equation*}
where $\mathcal{\tilde D}(\cdot, \cdot)$ is some arbitrary probability divergence and $\hat q (\omega) = \frac{1}{n_2}\sum_{i} I(\hat e_{i} = \omega)$.  In our experiments, we pick $\mathcal{\tilde D}$ to be the total variation norm $\norm{\cdot}_{TV}$.

A justification for this method is as follows. The distribution of $E_i$ can be described as a functional of the joint 
distribution: 
\al{
    q_0(E_i = e_i) &= \sum_{y_1}\sum_{y_2}p_0(y_1,y_2)I(e_i = y_2 - y_1)
}
While the model implied version is: 
\al{
    q_A(E^{(A)}_i = e_i) &= \sum_{y_1}\sum_{y_2}p_{MA}(y_2|y_1)I(e_i = y_2 - y_1)p_0(y_1).
}
These distributions $q_0$ and $q_A$ will be equal when the model is correct.  We instead assess the quality of the measurement assumption model by the closeness of $q_0$ to $q_A$. 

Note that  we use this conditional distribution in case the first score is not independent from whether a subsequent score is observed.  If a true generating model exists, then $ \mathcal{\tilde D}( q_0, q_A) = 0$, however, it is not necessarily the case that the true model would be the only minimizer. Moreover $q_A(E^{(A)}_i = e_i)$ is not a convex functional of $q_A$ in general, so we simply use this method to pick out the best choice of $A \in \mathcal{A}$. It is possible to get a precise estimate of $q_0$ as there are only $2N + 1$ categories, unlike if we tried to use the joint distribution directly as the size of the domain can be very large $(N+1)^2$, where $N \geq 30$ in many cognitive tests.  

As we will illustrate in Section \ref{sec:Simulations} in the simulations $\norm{\hat q_0 - \hat q_A}_{TV}$ will in general have very little dependence on the regularization parameter $\mu$. 

\subsection{Selection of $\mu$} \label{sec:reg_selection}

 In practice, we find that the procedure in Section \ref{sec:error_select} when using the total variation norm is relatively independent from the choice of $\mu$ (see Figure \ref{fig:intrinic_sim} as well as the supplementary material).  Therefore selection of $\mu$ requires a slightly different procedure. The main challenge is that for many samples, an estimate $\hat p(y|x)$ may end up outside of $\text{conv}(\Gamma)$, and thus the maximum likelihood estimate for the latent density is discrete \citep{Lindsay1983TheTheory}. Instead, we fit the latent model $p_M(\gamma|x)$ on the regular smoothed estimator, and select $\mu$ based on maximizing the two observation likelihood below
\begin{align}
    p_{\hat M(\mu)}(\cdot|x) &= \argmax_{p_M \in \mathcal{P}_{[0,1]}} \sum \hat p(y|x) \log\left(\int p_A(y|\gamma)p_M(\gamma|x)d\gamma\right)  + \mu \int \log\left(p_M(\gamma|x)\right) d\gamma \nonumber \\
    p_{\hat M(\mu)A}(Y_{1}, Y_2|X) &= \int p_A(Y_{1}|\gamma)p_A(Y_2|\gamma)p_{\hat M(\mu)}(\gamma|X)d\gamma \nonumber \\
    \ell_2(p_{\hat M(\mu)}) &= \frac{1}{n_2} \sum_{i = 1}^{n_2}\left[ \log(p_{\hat M(\mu)A}(Y_{i1}, Y_{i2}|X_{i}))\right] \label{eq:two_obs_likelihood}
\end{align}

In practice, $\hat p(y|x)$ is only based on the first observations.  As a means of smoothing the latent densities, we pick $\mu$ such that
\begin{align}
    \tilde \mu &= \argmin_{\mu} - \ell_2(p_{\hat M(\mu)}) .
\end{align}
If we only considered a training and test set, this would be a problem.  If $\hat p_1$ and $\hat p_2$ are the marginal conditional estimates of the distributions of $Y_{i1}$ and $Y_{i2}$ then these should converge to the same distribution.  Due to the non-identifiability of the problem, smoothing will only lower the value on the unseen likelihood as on a second sample when $\hat p_1 \approx \hat p_2$.  Instead we use smoothing and verify how well that the latent trait model fits the data where observations are generated as pairs from a single latent $\gamma$. This helps verify that the smoothed version will perform on unseen data, in a context of new samples tests from a cognitively stable individual.  In principal, we could fit the latent model using the pairs of observations, however, in our framework, this would require using a  high dimensional output classification algorithm with $(N+1)^2$ classes $\hat p(y_1,y_2|x)$.

\section{Simulations} \label{sec:Simulations}

\subsection{Intrinsic Variability Matching}

We illustrate the example of intrinsic variability matching as a method of selecting the correct model.  In this setting we must sample twice from the model observed distribution for a single latent variable sample $\gamma_i$. We repeat the following process to sample $100$ pairs of observed $Y$'s. 
\al{
    \gamma_i &\sim \text{Beta}(12,5)  &Y_{i1},Y_{i2}|\gamma_i \stackrel{iid}{\sim} MKM(K = \text{Gaussian}, h = 2, \gamma = \gamma_i) 
}

We repeat the sampling procedure $30$ times fitting the latent model on the first sample, and computing the total variation distance between the model implied intrinsic variability and the sample intrinsic variability.

\begin{figure}[htp!]
    \centering
    \includegraphics[width = 0.9\textwidth]{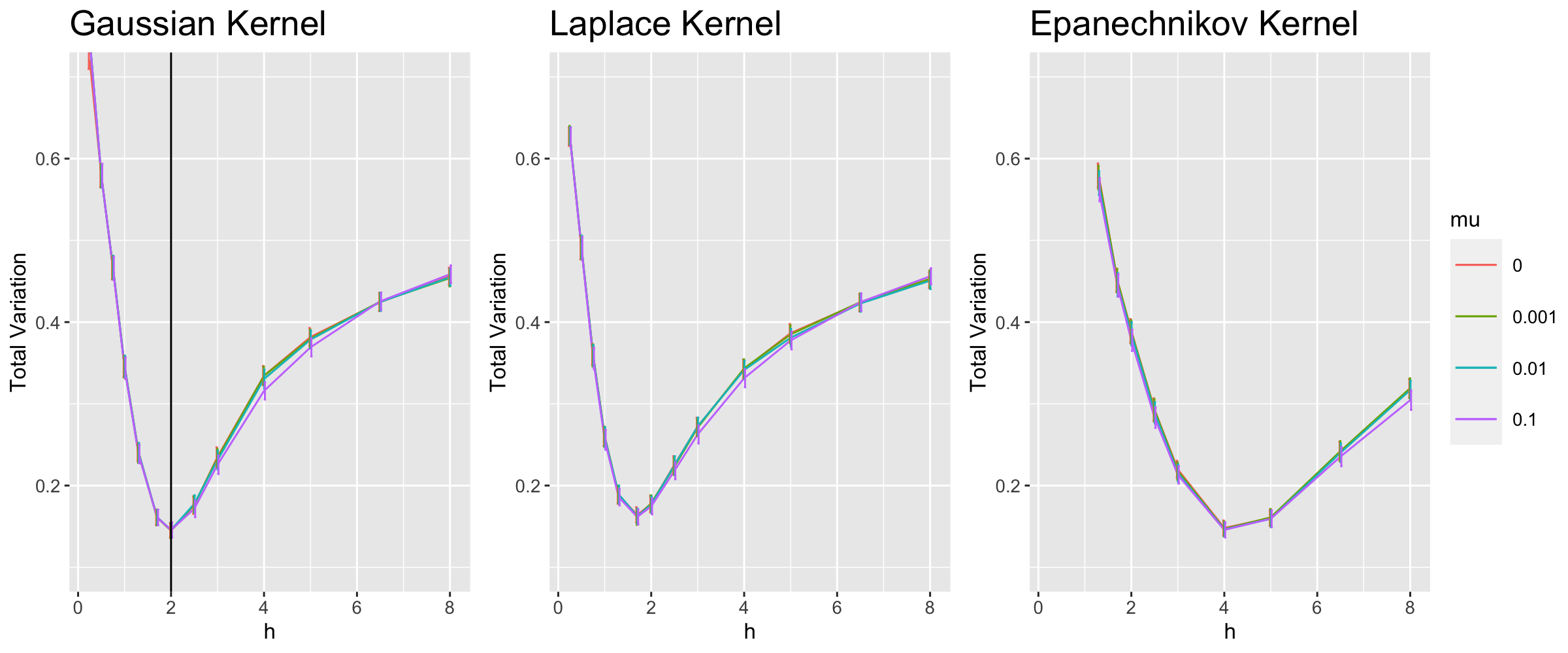}
    \caption{Intrinsic Variability Simulations.  Vertical black line denotes true data generating model.}
    \label{fig:intrinic_sim}
\end{figure}

We find in figure \ref{fig:intrinic_sim} this method tends to pick out the measurement assumption model quite effectively and the average total variation distance was smallest under the correct model. An interesting phenomena to observe is the fact that the regularization makes a relatively small impact on the intrinsic variability. This proves to be valuable in decoupling the selection of $\mu$ from $p_A$.

\subsection{Selecting the regularization tuning parameter}

After selecting a measurement assumption model we will select the regularization parameter. As illustrated in section \ref{sec:reg_selection}.  For a given $\mu$ we fit the latent model on $\hat p_1$, the empirical distribution of first scores. We then pick $\mu$ such that the two observation likelihood is maximized. We plot this procedure for 100 simulated datasets in Figure \ref{fig:reg_selection_sim}.

\begin{figure}[htp!]
    \centering
    \includegraphics[height = 0.25\textheight]{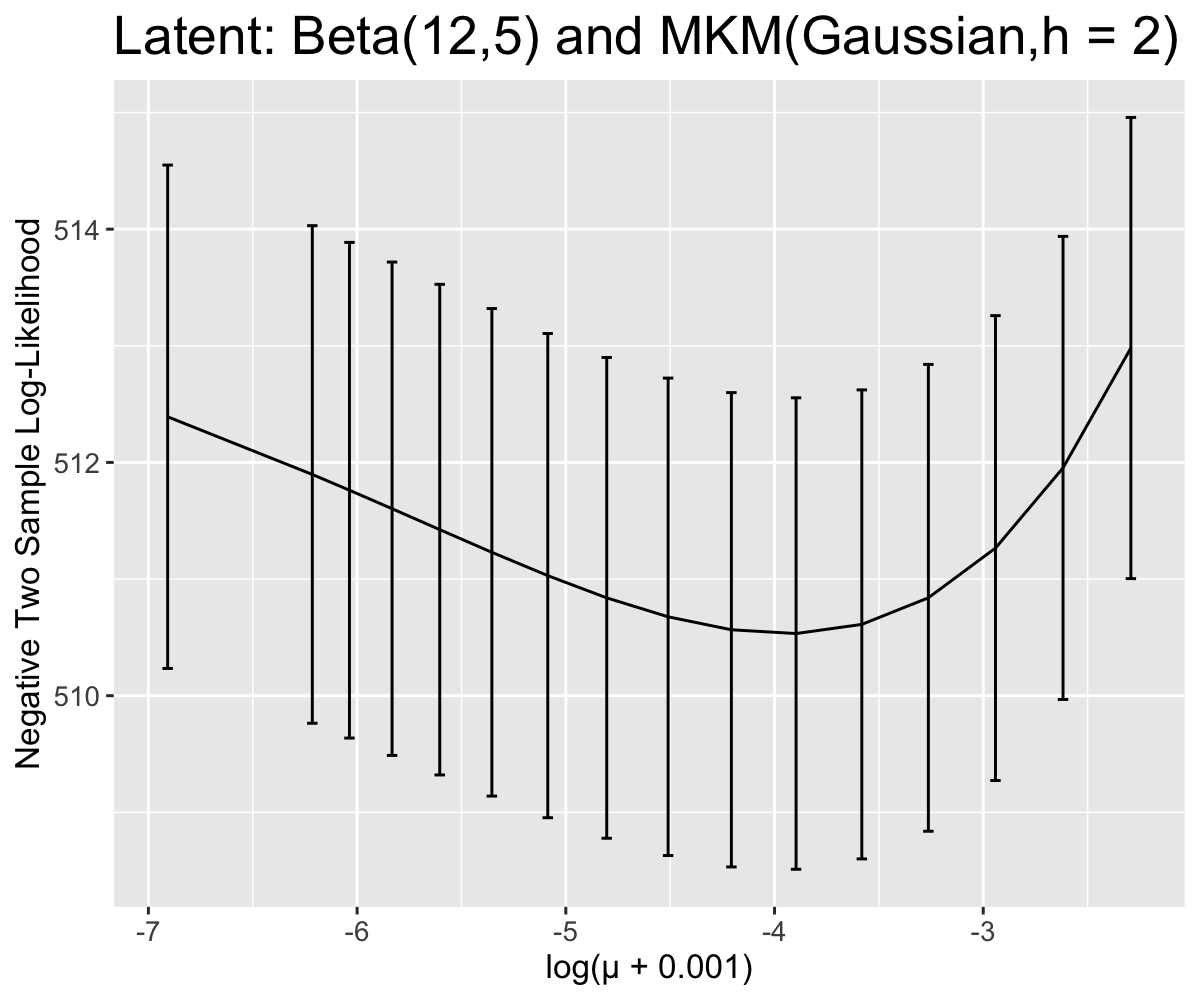}
    \caption{Regularization parameter $\mu$ selection. Including some regularization indeed seems to improve the performance on the second observation, with an optimal $\mu = 0.0193$ when results are averaged over all simulations.  The logarithm in the X-axis is natural logarithm.}
    \label{fig:reg_selection_sim}
\end{figure} 

We next turn to the main goal which is establishing an estimate for $\hat p(z|y)$ in the setting when $y,z$ are never observed together. 

\subsection{Performance of score conversion}

We now return to the main problem, estimating $p_0(z|y)$ in the case of a common latent quantile. We illustrate the performance on the selection procedure as a way to construct a good estimate of $p_0(y|x)$.  

We simulate a set of harmonizable variables under the following model. We suppose $F$ is the CDF of a $\text{Beta}(12,5)$ random variable and $G$ is the CDF of a $\text{Beta}(6,6)$.  Further we let $P_{A_y}(y|\gamma)$ be a $MKM(K = \text{Gaussian}, h = 2)$ model and $P_{A_z}(z|\zeta)$ be a $MKM(K = \text{Laplace}, h = 1)$.  
\al{
    \omega_i \sim U(0,1),& \quad \gamma_i = F^{-1}(\omega_i), \quad  \zeta_i = G^{-1}(\omega_i) \\
    Y_i \sim MKM(K = \text{Gaussian},& h = 2, \gamma = \gamma_i), \quad Z_i \sim MKM(K = \text{Laplace}, h = 1, \zeta = \zeta_i) 
}
Since the true model is known, we compute the population cross entropy as defined below and assess performance of an estimated $\hat p(z|y)$. 
\begin{align*}
    CE_{pop}(\hat p) &= \sum_{y = 0}^{N_y} p_0(y)\sum_{z = 0}^{N_z} -p_0(z|y)\log(\hat p(z|y)) = \sum_{y = 0}^{N_y} \sum_{z = 0}^{N_z} -p_0(y,z)\log(\hat p(z|y))
\end{align*}

We then pick a set of bandwidths and kernels and plot the population cross entropy. Since the true model is known, we do not include compact support kernels since except for extremely large bandwidths, we know that the population cross entropy will be infinite.  We first fit the latent model on $Y:$ $p_{M_Y}$ with the value of $\mu$ selected by our procedure.  We plot the cross entropy as a function of the bandwidth ($h_z$), kernel ($K_z$) and regularization parameter ($\mu_z$).  We also compare this conversion to the completely unregularized case ($\mu_Y = \mu_Z = 0$).  We plot these results in Figure \ref{fig:conversion_ce_sim} averaged over 10 monte-carlo simulations. 

\begin{figure}[htp!]
    \centering
    \includegraphics[height = 0.25\textheight, width = 0.8\textwidth]{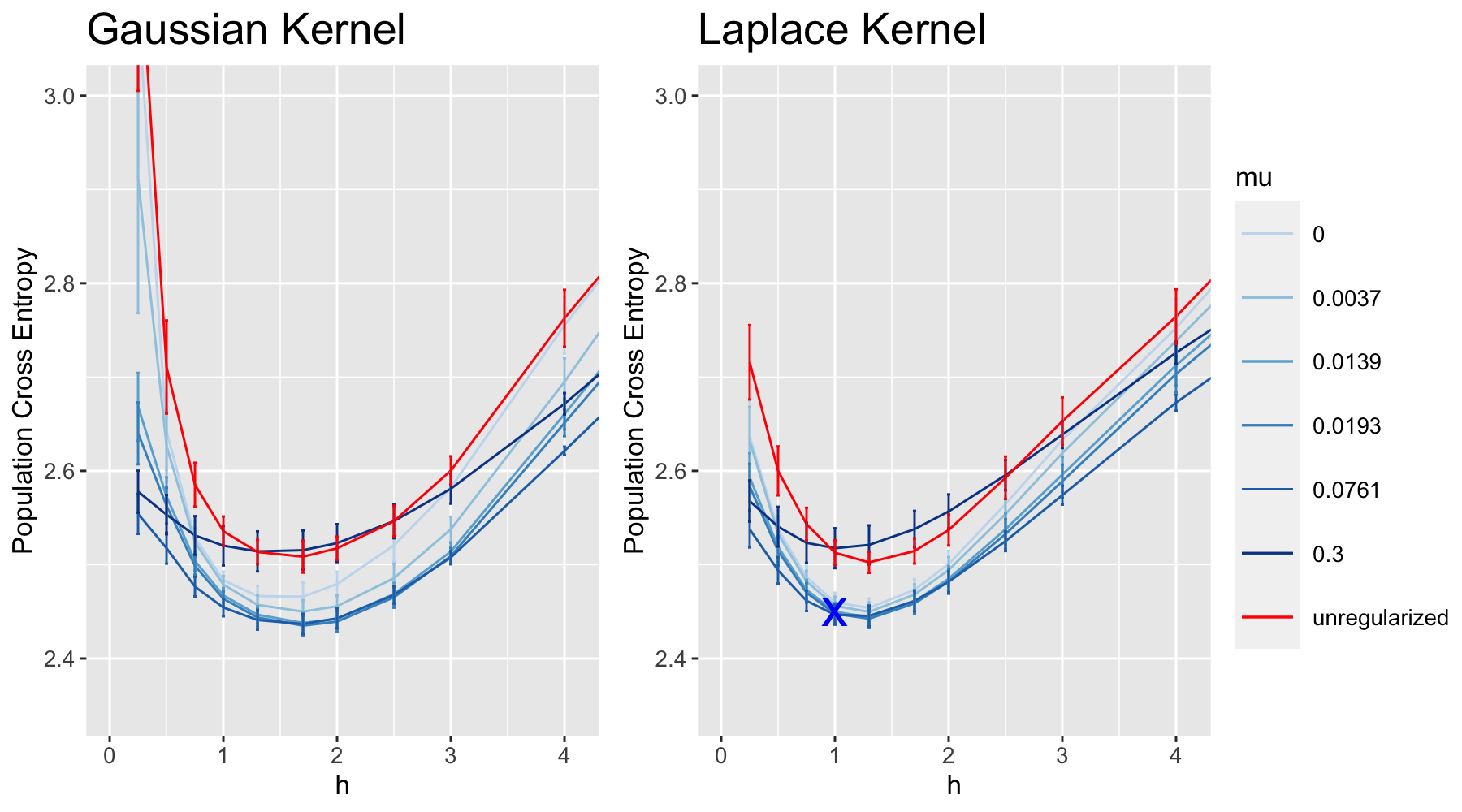}
    \caption{Conversion cross entropy as a function of model parameters and hyper-parameters for $Z$. Blue X indicates model selection based on intrinsic variability matching and tuning parameter selection of $\mu_z$}
    \label{fig:conversion_ce_sim}
\end{figure}

We firstly observe that the completely unregularized case results in a poor performance in the conversion. In this case $\hat \phi $ is a step function (Since must use the generalized inverse function for $\hat G^{-1}$), which may be associated with a poorer estimate of $p_0(z|y)$.  It is further apparent that if regularization is included only for the $Y$ branch, this partially alleviates the problem. We see however, in Figure \ref{fig:conversion_ce_sim} that a small amount of regularization on $Z$ improves the estimate of $\hat p(z|y)$, though the results are relatively robust to the choice of $\mu_z$ until $\mu_z$ is very large, after which the conversion performance suffers.

\subsection{Feasibility Tests}

We next simulated from the same model 300 times and plot the associated p-values from the first and second order feasibility tests of a set of measurement models with the associated kernels and bandwidths. We plot the first order feasibility test simulations in Figure \ref{fig:feasibility_sim1} and second order in Figure \ref{fig:feasibility_sim2}. 
\al{
    \gamma_i &\sim \text{Beta}(12,5) &Y_{i1},Y_{i2}|\gamma_i \stackrel{iid}{\sim} MKM(K = \text{Gaussian}, h = 2, \gamma = \gamma_i) 
}

\begin{figure}[h!]
\centering
\subfigure[First order feasibility tests]{
\includegraphics[width = 0.9\textwidth]{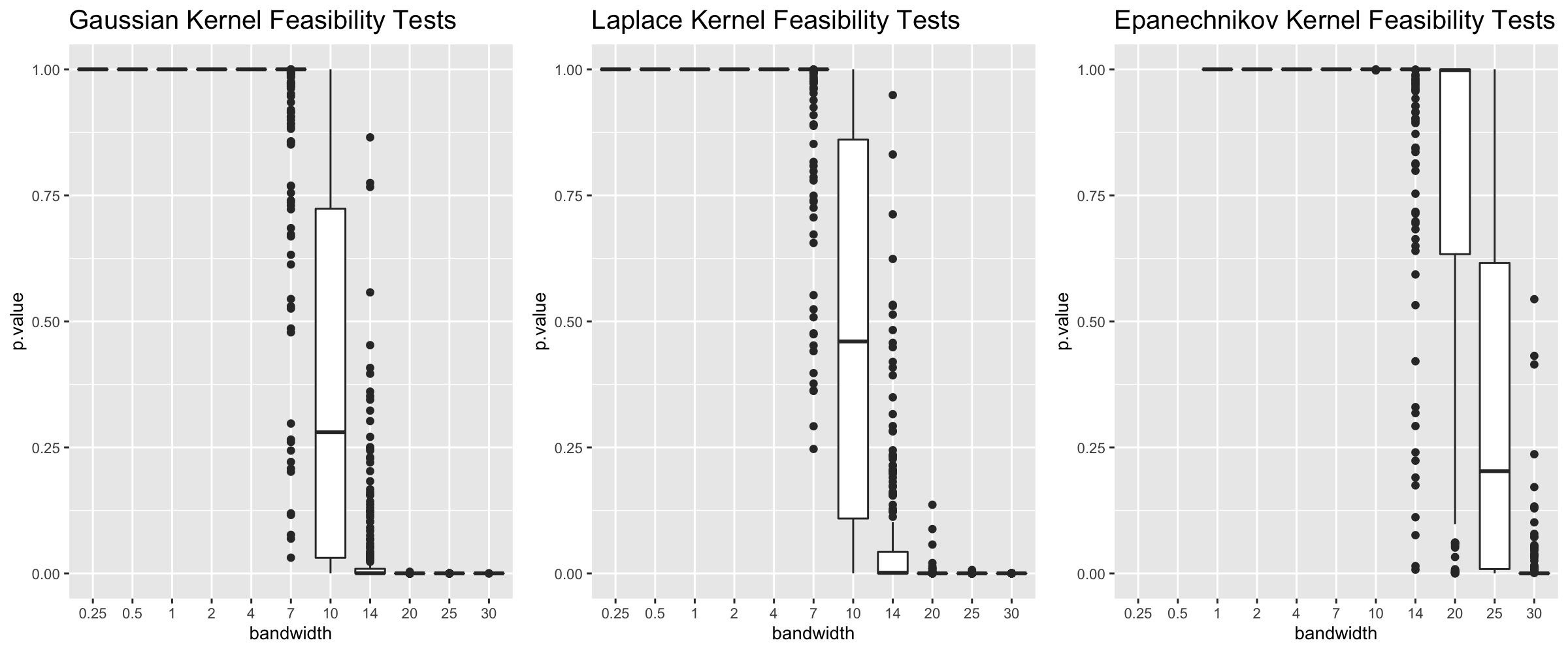}
\label{fig:feasibility_sim1}}
\\
\subfigure[Second order feasibility tests]{
\includegraphics[width = 0.9\textwidth]{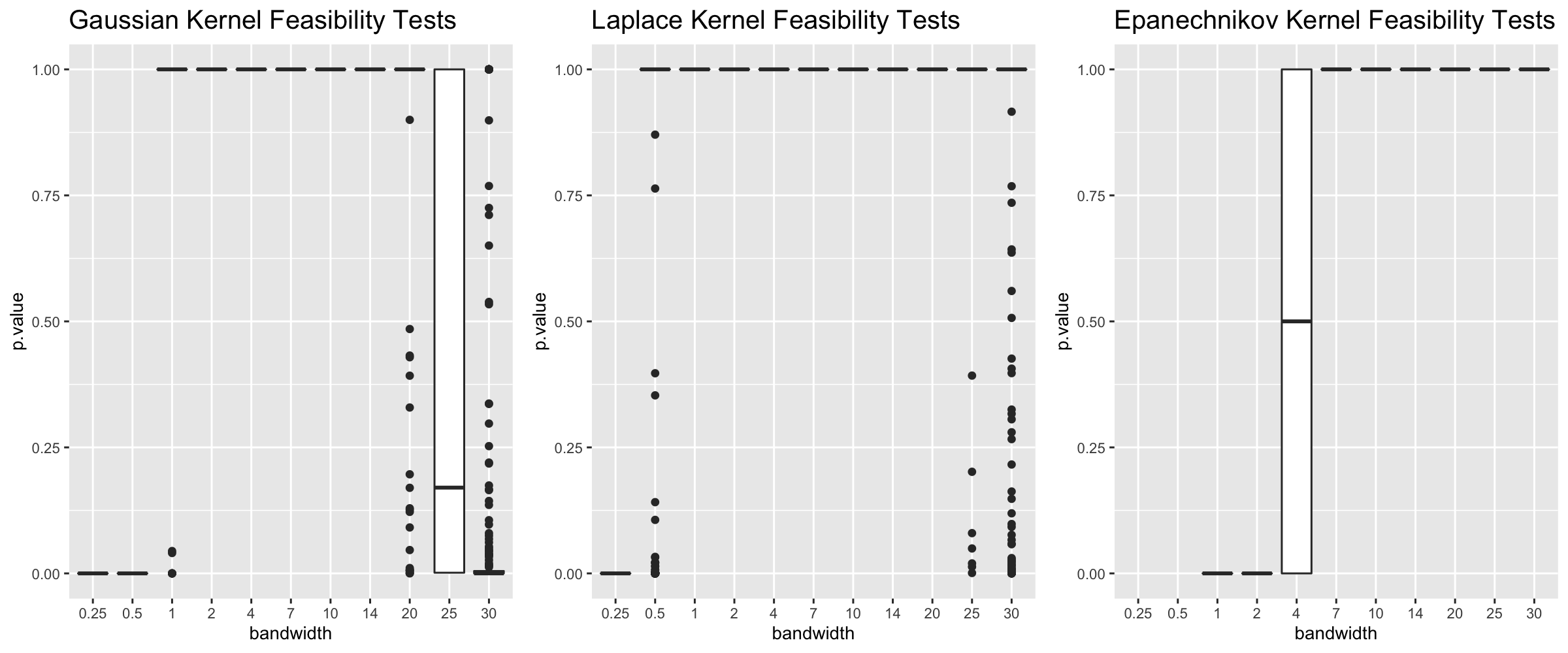}
\label{fig:feasibility_sim2}}
\caption{Feasibility test in simulations.  We find the first order test can discriminate when $h$ is too large, while the second order test can reject a model when $h$ is too small. }
\label{fig:feasibility_simulations}
\end{figure}

We note that there is a steep drop-off in the p-values as the bandwidth decreases.  We find that the first order feasibility test is able to pick out whether the bandwidth $h$ is too large, while the second order feasibility test detects when the bandwidth is too small.  This can be used in practice to select a narrow range for possible bandwidths in practice.  Measurement assumptions which are nearly correct however are impossible to distinguish with this test and all appear feasible for a moderate range of $h$. This method can therefore be used to detect violations of our assumption, however, we do note that many of these models may be similarly feasible.

\section{Application to the NACC data} \label{sec:application}
Our primary motivation has been to use this method for converting scores in the NACC Uniform dataset (UDS data freeze No. 47, Obtained July 2020). We consider the conversion between the proprietary C1 battery score, the Mini-Mental State Examination (MMSE) and the non-proprietary C2 battery score the Montreal Cognitive Assessment (MOCA). Both scores have a range of $\{0,\dots, 30\}$. For our study, we only include individuals who are cognitively normal as indicated by their global CDR (Clinical Dementia Rating) score of 0 and ages between 60 and 85 at the first visit.  We have a sample of $11194$ individuals having a recorded MMSE score and $6898$ with a MOCA. We consider a training set of first visits, as well as a validation set of second visits within 500 days of the first as a validation, intrinsic variability matching set.  We have $7614$ and $4051$ follow-up visits within each of these tests respectively.  Lastly we have $420$ individuals as part of the crosswalk dataset, a group of individuals with both scores measured \citep{Monsell2016ResultsStudy}.  Since the crosswalk dataset is comparably very small, learning the joint distribution is infeasible.  Instead we utilize the harmonizability assumptions which allow us to use the whole training dataset to estimate the latent distribution. We reserve the crosswalk dataset to verify the performance of converting scores. 

A simple conditional regression model is used to estimate the conditional density as a function of age, sex and education level. We discretize education level to $<16$ and $\geq 16$ years which represents approximately whether or not the individual has received a college degree, leading to 4 categories within the population. We then use a simple conditional distribution estimator for an age range $w$:  

\al{
    \hat p(y|x) &= \frac{\sum_{i = 1}^n I(y = Y_i)I(|x_{age} - X_{age, i}| \leq w)I(x_{group} = X_{group, i})}{\sum_{i = 1}^n I(|x_{age} - X_{age, i}| \leq w)I(x_{group} = X_{group, i})}
}

Using our model selection procedure outlined in section \ref{sec:error_select} we select the conditional models $P_A$ for each score. By our method we select the binomial model for the MMSE and the Laplace $h = 1.34$ model for the MOCA model, though the Gaussian and binomial models produced achieved very similar distance to the intrinsic variability distributions for the MOCA.  We find all 4 of the above models obtain a p-value of $1$ for the feasibility tests. Using the procedure outlined in section \ref{sec:reg_selection} we select $\mu_Y = 0.0055, \mu_Z = 0.0113$ for our model.  Further details are included in the supplementary material section \ref{supp:sec:nacc_analysis}.

{\bf Validation based on the crosswalk study.}
To validate our method against alternatives, we consider a prediction problem on the crosswalk study, a small group of 420 cognitively normal individuals who completed both the MOCA and MMSE cognitive tests \citep{Monsell2016ResultsStudy}. We label these observations $i \in \{1, \dots, n_{cw}\}$ with observations $(X_{i}^{cw}, Y_i^{cw}, Z_i^{cw})$

We use this dataset as a test set our methods against the existing procedure.  The current standard method for data harmonization is a simple normal $Z$-score matching procedure
\al{
    \hat Z &= \frac{\hat \sigma_Z}{\hat \sigma_Y}(Y - \hat \mu_Y) + \hat \mu_Z 
}
where $\hat \mu$ and $\hat \sigma$ are the sample mean and standard deviation estimates from the corresponding $Y$ and $Z$ scores.  This is also rounded to the nearest score in practice. Though this conversion is deterministic, we can introduce a minor modification to compare it to our estimates of $p_0(z|x,y)$.  Since this method assumes normality of the score distributions, we add normal random noise to $\hat Z$, $N(0,\hat \sigma^2_Z)$ and round to the nearest score. 

Although our goal is to produce an estimate of $p_0(z|x,y)$, the small crosswalk dataset for which we have complete observations from $(X_i, Y_i, Z_i)$ is insufficient to construct any precise estimate of $p_0(z|x,y)$. This is what necessitated our original assumptions of harmonizable observations $Y,Z$ allowing us to utilize the larger data sets that only include a single score $Y$ or $Z$.   We used the learned $p_{\vect{M}(\vect{\mu}), \vect{A}}(z|y,x)$  based on learning the latent marginals and converting via our mapping $\phi_{\gamma \to \zeta}$.  We  use this crosswalk data as a validation set to assess the quality of the score conversion
using the sample cross entropy ($CE$)
\begin{equation} \label{eq:ce}
    CE(p_{\vect{M}(\vect{\mu}), \vect{A}}) = -\sum_{i = 1}^{n_{cw}} \log\left( \int p_{A_2}(Z_i^{(cw)}|\phi_{\gamma \to \zeta, (\vect{\mu})}(\gamma|X_i^{(cw)}))\frac{p_{A_1}(Y_i^{(cw)}|\gamma) p_{M_1(\mu_1)}(\gamma|X_i^{(cw)})}{p_{M_1(\mu_1) A_1}(Y_i^{(cw)}|X_i^{(cw)})} d\gamma\right) 
\end{equation}
where $\vect{M}(\vect{\mu}), \vect{A}$ indicate that this conversion depends both on both branches of the models. In practice, we can compute the integral in equation \eqref{eq:ce} using a binned approximation of the integral.

{\bf Performance on the NACC data.}
We next outline the performance of the conversion estimator.  We select the model on $Y$ using the previous techniques and find binomial measurement assumption model with $\mu_Y = 0.0055$ is optimal.  For the $Z$ model we select a Laplace kernel with a bandwidth of $h = 1.34$ and $\mu_Z = 0.0113$ to be optimal. Once again to display the effect of the measurement model on conversion, we fix the measurement model and smoothing parameter on $Y$, then pick a variety of models on $Z$ with different regularization parameters $\mu_Z$ and measurement assumption models $p_{A_Z}(z|\zeta)$.  In the supplementary material, we also introduce a parametric model to compare along with the modified Z-score matching method, and the completely unregularized models where $\mu_Y = 0, \mu_Z = 0$.  

\begin{figure}[htb!]
    \centering
    \includegraphics[height = 0.25\textheight]{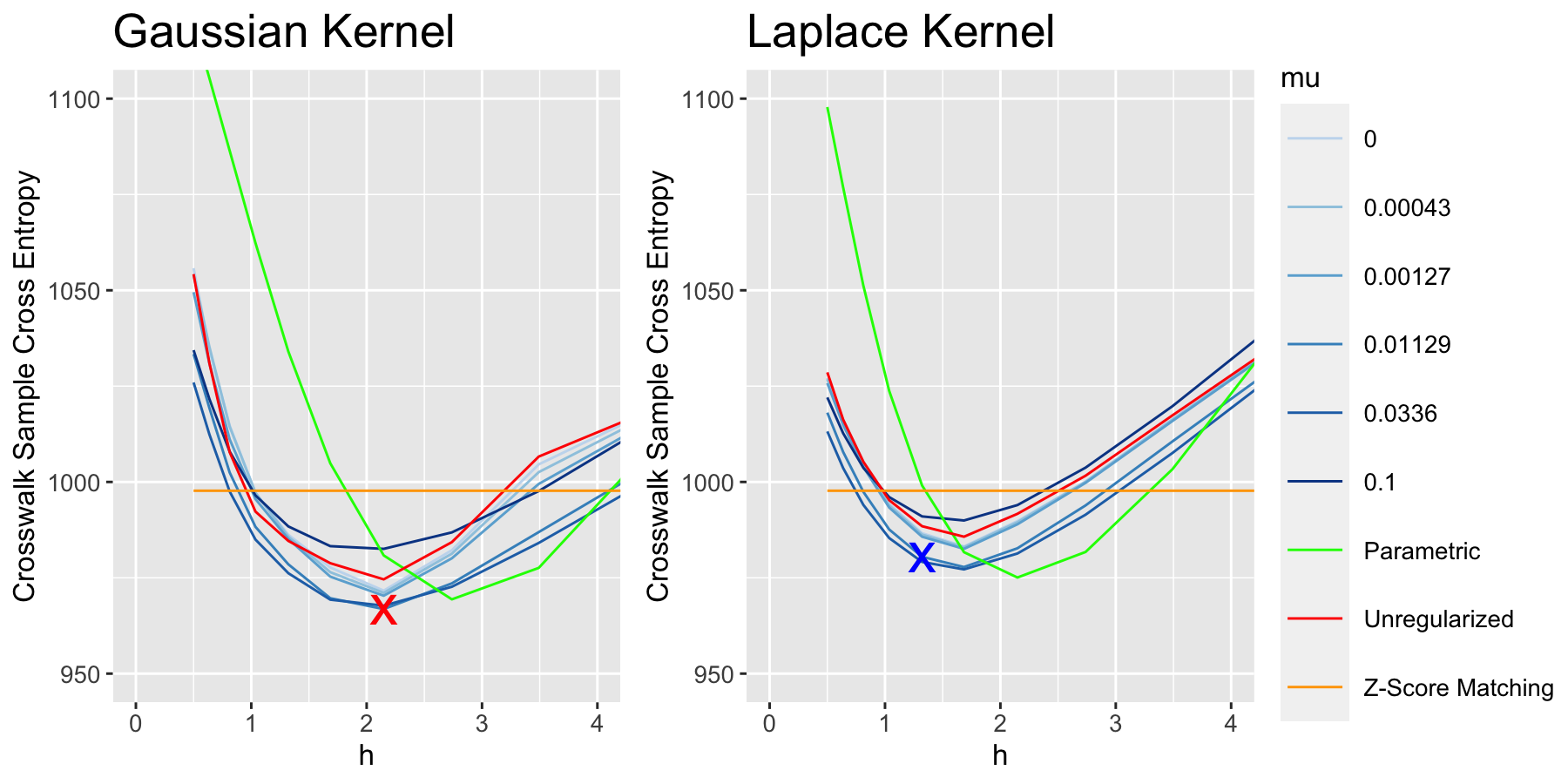}
    \caption{Conversion from MMSE to MOCA as a function of the parameters $\mu_Z$ as well as the measurement assumption model $p_{A_Z}$.  The blue cross indicates the model chosen by our procedure and the red cross indicated the empirical optimal model for conversion. }
    \label{fig:MMSE_MOCA_conversion}
\end{figure}

From figure \ref{fig:MMSE_MOCA_conversion}, we find that our method of selecting a model (blue cross, cross entropy: $980.4$) is performs quite well, and the model with the lowest cross entropy of all was the Gaussian $h = 2.14$, $\mu = 0.0113$ (red cross, cross entropy: $966.8$).  

We notice that the non-parametric approach outperforms the parametric and naive Z-score methods for the particular conditional model selected, as well as for the overall best model for converting scores. The parametric model does not perform optimally under the selected measurement model, but best under more dispersed models (large $h$). Additionally, we find that with a small amount of regularization, our method outperforms the un-regularized model and the conversion is relatively robust to the choice of regularization parameter, until a very large value is chosen ($\mu_Z \geq 0.1$). 

In applications where one can observe both scores, selection of $\vect{\mu}$ and $p_{\vect{A}}$ can instead be done through performance on the conversion task.  This analysis provides a justification for our method of selecting these models without these complete cases, so long as the harmonizbility assumption is reasonable.

\section{Discussion}

In this paper we introduced a framework for the problem of data harmonization, inspired by a problem in Alzheimer's and dementia research.  We connect the problem to the existing research in non-parametric mixing density estimation.   

It may be true independence assumption $A\indep (\Omega, Y,Z)|X$ is not satisfied in practice, so the latent distribution would be of $p_{M_Y}(\gamma|x, a = 0)$ and  $p_{M_Z}(\zeta|x, a = 1)$ that depends in $a$.  In principal, which test is observed may depend on the relative quantile $\Omega$ of an individual. At the NACC data, the change in tests were due to switching from a proprietary to non-proprietary methods, so we do not believe this to be an issue, as this is independent from the cognitive ability of the participants. Though in principal one could develop a simple sensitivity analysis procedure.  If $a$ is some offset value one could use the conversion function: $\phi_{sens. \gamma \to \zeta}(\gamma|x) = G^{-1}(F(\gamma|x) + a|x)$ though one would have to be mindful of the boundary conditions for the selected values of $a$. 

By theorem \ref{thm:cdf} the our harmonizability assumption implies an exact form of the conversion between latent variables.  This assumption is strong in that the two latent variables are actually measuring the exact same trait.  In reality, these traits may not be the deterministically related, but rather highly correlated.   In this case, the mapping we are learning is really the optimal transport map (under the $c(x,y) = |x - y|$ cost function) between the continuous latent trait models $p_{M_Y}$ and $p_{M_Z}$ \citep{Villani2003TopicsTransportation}. 
Even under misspecification of this assumption, this may help explain good performance of converting scores on real world data. 

Our methods draw similarities to semi-parametric factor analysis \citep{Joreskog2001FactorApproaches, Gruhl2013AVolumes}, and item response theory \citep{Johnson2007ModelingSplines,Woods2006ItemDensities,Paganin2021ComputationalModels}. These approaches often use parametric models for a latent trait distribution, and a non-parametric link function, relating the observed data to the latent data (or vice versa). Our approach is more flexible that we allows a non-parametric model of the latent trait distribution $p_M$, however this requires specification of how a conditional distribution of the observed variables given the latent trait $p_A(y|\gamma)$.  As is always the case in nonparametric mixture models we have a non-identifiability problem when the support of the observed distribution is finite \citep{Lindsay1995MixtureApplications}. In a more theoretical study of a similar problem,  \cite{Tian2017LearningParameters} and \cite{Vinayak2019MaximumParameters} proved that the NPMLE and a moment matching method recovers the true latent distribution within $\mathcal{O}_P(\frac{1}{N})$ in terms of the Wasserstein 1 distance \citep{Tian2017LearningParameters, Vinayak2019MaximumParameters} in the case where the conditional distribution $p_A$ is a binomial distribution.

We present a preliminary approach to deal with this particular structure of missing data problem. However there are many extensions of interest.  In particular to longitudinal and multivariate data harmonization, as well as theoretical questions involving recovering a latent distribution $P_M$ under other conditional models $p_A$.  Such a model will be extremely important for further statistical applications such as imputation of missing scores, change point detection of mild cognitive impairment and classification of neuro-degenerative disease via mixture models on the test scores.



\bibliography{references}{}
\bibliographystyle{imsart-nameyear}

\begin{appendix}

\end{appendix}

\section*{Acknowledgements}
The authors would like to thank the anonymous referees, an Associate
Editor and the Editor for their constructive comments that improved the
quality of this paper.

The authors were partially supported by US NIH Grants U01 AG016976 and U24 AG072122).
The second author was partially supported by US NSF Grants  DMS-195278 and DMS-2112907.
The third author was partially supported by US NSF Grant DMS-1711952.


\begin{supplement} \label{supp}
\stitle{Supplementary Material}
\sdescription{Proofs of all theorems,additional computation details, additional simulations and further real data analysis.} 

\end{supplement}

\newpage

\begin{frontmatter}
\title{Supplement to: Data Harmonization Via Regularized Nonparametric Mixing Distribution Estimation}
\runtitle{Data Harmonization}

\begin{aug}
\author[A]{\fnms{Steven} \snm{Wilkins-Reeves}\ead[label=e1]{stevewr@uw.edu}},
\author[A]{\fnms{Yen-Chi} \snm{Chen}\ead[label=e2,mark]{yenchic@uw.edu}}
\and
\author[B]{\fnms{Kwun Chuen Gary} \snm{Chan}\ead[label=e3,mark]{kcgchan@uw.edu}}
\address[A]{Department of Statistics,
University of Washington,
\printead{e1,e2}}

\address[B]{Department of Biotatistics,
University of Washington,
\printead{e3}}
\end{aug}

\end{frontmatter}

Code for all the figures and simulations is available at \url{https://github.com/SteveJWR/Data-Harmonization-Nonparametric}. \\

\section{A Parametric Model for the Latent Mixture} \label{sec:par_model}

As an easy to compute alternative to the non-parametric mixture model, we also introduce a parametric model. \\

The parametric model assumes $p_M(\gamma|x) = p_M(\gamma|x, \theta)$, where $\theta \in \Theta$ is the underlying parameter.
A straightforward approach to estimate $\theta$ is via the maximal likelihood estimator (MLE). 
The observed data log-likelihood and its expectation are  
\al{
    \ell_n(\theta) &= \frac{1}{n}\sum_{i = 1}^n\Big[\log\Big(\int p_A(Y_i|\gamma)p_M(\gamma|X_i, \theta) d\gamma\Big)\Big] \ ,\\
        \ell(\theta) &= \E_{Y,X}\Big[\log\Big(\int p_A(Y|\gamma)p_M(\gamma|X, \theta) d\gamma\Big)\Big]
}
respectively.  


The MLE finds an estimator of $\theta$ by maximizing $\ell_n(\theta)$.
However, $\ell_n(\theta)$ may be difficult to maximize
for many parametric models $p_M(\gamma|x, \theta)$
In general, one could estimate this using the EM algorithm \citep{Dempster1977MaximumAlgorithm} or numerical approximations,
but this may be time-consuming and may not find the global maximum.


To deal with this problem,
we propose a variational approach
to the above likelihood function.
Let $r(\gamma| Y) = \frac{p_A(y|\gamma)}{\int p_A(y|\gamma) d\gamma}$.
We can rewrite the likelihood as
\al{
    \ell(\theta) &=  \E_{Y,X}\Big[\log\Big(\int p_A(Y|\gamma)p_M(\gamma|X, \theta) d\gamma\Big)\Big] \\
    &=  \E_{Y,X}\Big[\log\Big(\int r(\gamma| Y) p_M(\gamma|X, \theta) d\gamma\Big)\Big] -  \E_{Y,X}\Big[\log\Big(\int p_A(Y|\gamma) d\gamma\Big)\Big] \\
    \text{By Jensen's Inequality } &\geq \E_{Y,X}\Big[\int r(\gamma| Y)\log\Big( p_M(\gamma|X, \theta)\Big) d\gamma\Big] -  \E_{Y,X}\Big[\log\Big(\int p_A(Y|\gamma) d\gamma\Big)\Big] := \widetilde \ell(\theta).
}
Similarly, we define $\widetilde \ell_n(\theta)$ as a lower bound of $\ell_n(\theta)$. This lower bound function $\tilde \ell(\theta)$ is easy to optimize if $p_M(\gamma|x, \theta)$ is log concave in $\theta$ for each $x, \gamma$.

In this paper, we consider the logit-normal model:
\begin{align*}
    p(\gamma|x; \theta) &= \frac{1}{\gamma(1-\gamma)}\frac{\lambda}{\sqrt{2\pi}}\exp\left(-\frac{\lambda^2}{2}(\logit(\gamma) - \beta^T x)^2\right),
\end{align*}
where $\theta = (\beta, \lambda) \in \mathbb{R}^{d} \times (0,\infty)$.
This logit-normal model has closed form solutions for $\widetilde \theta = \argmax_{\theta} \widetilde{\ell}(\theta)$ and $  \widetilde \theta_n = \argmax_{\theta} \widetilde{\ell}_n (\theta) $ as given by the following theorem:

\begin{supptheorem}\label{supp_thm:par_closed_form}
    The approximated latent model with sample likelihood $\tilde \ell_n(\theta)$ and population likelihood $\tilde \ell(\theta)$ have a unique closed form solution. 
    \begin{align}
        \tilde \beta_n &= \sum_{i = 1}^n[X_i X_i^T]^{-1}  \sum_{i = 1}^n\bigg[ X_i \Big[\int r(\gamma|Y_i)\logit(\gamma)d\gamma \Big] \bigg] \\
        \tilde \lambda_n &= \Big\{\frac{1}{n}\sum_{i = 1}^n\Big[\int r(\gamma|Y_i) (\logit(\gamma) - \tilde \beta_n^T X_i)^2 d\gamma \Big] \Big\}^{1/2}
    \end{align}
    With corresponding population versions of $\tilde \beta, \tilde \lambda$
        \begin{align}
        \tilde \beta &= [\E[X X^T]]^{-1} \E\left[X\Big(\int r(\gamma| Y)\logit(\gamma)d\gamma \Big) \right] \\
        \tilde \lambda &= \Big\{\E\Big[\int r(\gamma| Y) (\logit(\gamma) - \tilde \beta^T X)^2 d\gamma \Big] \Big\}^{1/2}
    \end{align} if the following hold: 
    \begin{enumerate}
        \item $P[\norm{X}_2] < \infty$ for each $j$\\
        \item $P[\norm{X X^T}_2] < \infty$ where $\norm{\cdot}_2$ applied to a matrix denotes the spectral norm \\
        \item For some $b > 0$,  $b \leq \inf_{y \in \{0,1,\dots, N\}}\int p_A(y|\gamma) d\gamma$ 
        \item For some $y^*$ $r(\gamma|y^*)  >  0  $ for $\gamma \in \Lambda  \subset [0,1]$ where $\Lambda$ is some subset of Lebesgue measure greater than $0$ and $P(Y = y^*) > 0$ 
    \end{enumerate}
\end{supptheorem}

See the proof in the section \ref{sec:closed_form_proof}.

\begin{supptheorem} \label{supp_thm:ral}
    Under the same mild conditions outlined in Section \ref{sec:ral_proof}, 
    the estimates $\tilde \beta_n, \tilde \lambda_n$ are consistent, regular asymptotically linear estimators of $\tilde \beta, \tilde \lambda$.
\end{supptheorem}

See the proof in the section \ref{sec:ral_proof} and the assumptions. 

While the above variational estimator is not the actual MLE,
it has a closed-form. 
Moreover, it is similar to the least squares solution, with replacing $Y$ with $\int r(\gamma| Y)\logit(\gamma)d\gamma$.  This method can either be used on extremely large datasets where computational cost is an issue, or as an initial estimate for the following non-parametric estimator.

\section{Proofs of Supplemental Theorems}

\subsection*{Proof of Supplemental Theorem \ref{supp_thm:par_closed_form}} \label{sec:closed_form_proof}

\begin{proof}

Let $ \gamma^\dagger|Y$ be a random variable with a conditional PDF $r(\gamma| Y)$ that is not necessarily to be  the actual PDF.  We use the dagger, as this may not be the true distribution of $\gamma$, but reflects this conditional, under the case where we assume a latent uniform distribution on $\gamma$.

We first note that as $\beta_j \to \pm \infty$ then $p(\gamma|x; \theta) \to 0$ for all $\gamma \not \in \{0,1\}$ and hence the lower bound likelihood $\tilde \ell(\theta) \to -\infty$.  Furthermore, if $\lambda \to 0$ then $p(\gamma|x; \theta) \to 0$ for all $\gamma \not \in \{0,1\}$ and if $\lambda \to \infty$ then $p(\gamma|x; \theta) \to 0$ for all $\gamma \not = \beta^T x$.  Therefore, $\tilde \ell(\theta)$ is coercive in $\theta$ and any maximizer must on a compact subset $\Theta_c \in \R^d \times (0,\infty)$. \\
We then assume the following conditions on $p_A$ and $P_{X,Y}$
\begin{enumerate}
    \item $P[\norm{X}_2] < \infty$ for each $j$\\
    \item $P[\norm{X X^T}_2] < \infty$ where $\norm{\cdot}_2$ applied to a matrix denotes the spectral norm \\
    \item For some $b > 0$,  $b \leq \inf_{y \in \{0,1,\dots, N\}}\int p_A(y|\gamma) d\gamma$ 
    \item For some $y^*$ $r(\gamma|y^*)  >  0  $ for $\gamma \in \Lambda  \subset [0,1]$ where $\Lambda$ is some subset of Lebesgue measure greater than $0$ and $P(Y = y^*) > 0$ 
\end{enumerate}
The fourth condition is a very mild one which is required for the coersiveness of $\ell(\theta)$.  I.e. as long as there is some $y$ such that $r(\gamma|y^*)$ is positive in a region of positive Lebesgue measure. The first three conditions will be utilized for a dominated convergence theorem argument. Note that condition 5 implies $r(\gamma|y^*) \leq \frac{1}{b}$.  Let $W = (X,Y)$ and let $\ell(\theta) = P(f(W; \theta))$. 

\al{
    \text{since }  \frac{d f(W; \theta)}{d\beta} &= \int r(\gamma| Y) \frac{d (- \log(p(\gamma|X, \theta)))}{d\beta} d\gamma  \\
    &= \int r(\gamma| Y) 2\lambda^2\{X X^T \beta  -  \logit(\gamma)X\} d\gamma  \\
    P\left[  \norm{\frac{d f(W; \theta)}{d\beta}}_2\right] &\leq \frac{1}{b}  \left(\sup_{\lambda \in \Theta_c} 2 \lambda^2 \sup_{\beta \in \Theta_c} \norm{\beta}_{2} P\left[ \norm{X X^T}_2\right]   + \int \logit(\gamma)d\gamma  P\left[ \norm{X }_2\right]\right) < \infty \\
    \text{and }  \frac{d f(W; \theta)}{d\lambda} &= \int r(\gamma| Y) \frac{d (- \log(p(\gamma|X, \theta)))}{d\lambda} d\gamma  \\
    &= \int r(\gamma| Y) \{ - \frac{1}{\lambda}  + \lambda(\logit(\gamma) - \beta^T X)^2\} d\gamma   \\
    P\left[ \left|\frac{d f(W; \theta)}{d\lambda}\right|\right]&\leq  \frac{1}{b}\Bigg\{ \sup_{\lambda \in \Theta_c}\frac{1}{\lambda}  + \sup_{\lambda \in \Theta_c} \lambda \bigg(\int|\logit(\gamma)|^2d\gamma \\
    &+ 2\int|\logit(\gamma)|d\gamma \sup_{\beta \in \Theta_c} \norm{\beta}_2 P\left[ \norm{X }_2\right] +  \sup_{\beta \in \Theta_c} \norm{\beta}^2_2 P\left[ \norm{XX^T }_2\right]\bigg)^2\Bigg\} d\gamma  < \infty\\
}
Hence by the dominated convergence theorem, we can swap the derivative and expectation.

The rest follows by setting the derivatives to 0:
\al{
    \frac{d \widetilde \ell(\theta)}{d\theta} &= \E_{X, Y}\Big[\int r(\gamma| Y) \frac{d (- \log(p(\gamma|X, \theta)))}{d\theta} d\gamma \Big] \\
    \frac{d \widetilde \ell(\theta)}{d\beta} &= \E_{X, Y}\Big[\int r(\gamma| Y) 2\lambda^2\{X X^T \beta  -  \logit(\gamma)X\} d\gamma \Big] = 0  \\
    \implies
    \E_{X, Y} \Big[\int r(\gamma| Y)dy  X X^T \Big]  \tilde \beta
    &= \E_{X, Y}\bigg[ X \Big[\int r(\gamma| Y)\logit(\gamma)d\gamma \Big] \bigg] \\
    \implies \tilde \beta &= \E_{X}[X X^T]^{-1}  \E_{X, Y}\bigg[ X \Big[\int r(\gamma| Y)\logit(\gamma)d\gamma \Big] \bigg] \\
    \frac{d \widetilde \ell(\theta)}{d\lambda} &= \E_{X, Y}\Big[\int r(\gamma| Y) \{ - \frac{1}{\lambda}  + \lambda(\logit(\gamma) - \beta^T X)^2\} d\gamma \Big]  = 0 \\
    \implies \frac{1}{\tilde \lambda^2} &= \E_{X, Y}\Big[\int r(\gamma| Y) (\logit(\gamma) - \tilde \beta^T X)^2 d\gamma \Big] 
}
And similarly we have sample optima 
\al{
    \widehat \beta &= [\frac{1}{n}\sum_{i = 1}^n X_i X_i^T]^{-1}  \frac{1}{n}\sum_{i = 1}^n\bigg[ X_i \Big[\int r(\gamma|Y_i)\logit(\gamma)d\gamma \Big] \bigg] \\
    \frac{1}{\widehat \lambda^2} &= \frac{1}{n}\sum_{i = 1}^n\Big[\int r(\gamma|Y_i) (\logit(\gamma) - \widehat \beta^T X)^2 d\gamma \Big] 
}

\end{proof}

\subsection*{Proof of Supplemental Theorem \ref{supp_thm:ral}} \label{sec:ral_proof}

\begin{proof}
First Assume the moment conditions in Theorem \ref{supp_thm:par_closed_form}.  We will study this as a standard $M$-estimation problem and use
Theorems 5.14 and 5.41 of \citep{Vaart1998AsymptoticStatistics} to prove the asymptotic linearity of the estimator. Let $M(\theta) = -\widetilde \ell (\theta), \thickspace M_n(\theta) = -\widetilde \ell_n (\theta)$

Let $m_\theta = -\ell_1(\theta)$ where
\al{
    -\ell_1(\theta)(x,y) &= \int r(\gamma| y) \{-\log(\lambda) + \frac{\lambda^2}{2}(\logit(\gamma) - \beta^T x)^2\} d\gamma \\
    &\geq -\log(\lambda)  + \frac{\lambda^2}{2} \int r(\gamma| Y) \{\logit(\gamma) - \int r(\gamma| Y)\logit(\gamma)d\gamma \}^2 dy,
}
where $r(\gamma| Y) = \frac{p_A(y,\gamma)}{\int(p_A(y,\gamma) d\gamma)}$. Noting that $\int|\logit(\gamma)|d\gamma, \int(\logit(\gamma))^2d\gamma < \infty$.  Then for all $y \in \{0, \dots, N \}$
\al{
    0 &< \int r(\gamma| Y) \{(\logit(\gamma) - \int r(\gamma| Y)\logit(\gamma)dy )^2\} d\gamma < \infty \\
    \text{Let } B &= \inf_{y \in \{0, \dots, N\}} \int r(\gamma| Y) \{(\logit(\gamma) - \int r(\gamma| Y)\logit(\gamma)d\gamma )^2\} d\gamma .
}

Hence 
\al{
    -\ell_1(\theta)(x,y) &\geq -\log(\lambda) + \frac{\lambda^2}{2}B \\
    &\geq - \frac{1}{2}\log(\frac{1}{B}) + \frac{1}{2} \\
    &> -\infty
}
and
$m_\theta$ is continuous and uniformly bounded. Additionally, since  $\lim_{\lambda \to 0,\infty} -\log(\lambda) + \frac{\lambda^2}{2}B = \infty$ and $\lim_{\beta_k \to \pm \infty} -\ell_{1}(\theta) = \infty$, we can consider the maxima over the compact extension of our space $\overline \Theta = \{ \mathbb{\overline R}^d \times [0,\infty]\}$ where $ \mathbb{\overline R} = [-\infty, \infty]$. 

By theorem 5.14 \citep{Vaart1998AsymptoticStatistics}, for all $\epsilon > 0$
\al{
    P(\norm{\widehat \theta - \theta_0}_2 \geq \epsilon) \to 0
}
and $\widehat \theta$ is consistent for $\theta_0$.

Next we show asymptotic linearity for a given consistent estimator by theorem 5.41 in \citep{Vaart1998AsymptoticStatistics}. Let 

\al{
    \psi(x,y;\theta) = -\int r(\gamma| Y) \frac{d \log(p(\gamma|x,\theta))}{d\theta} d\gamma.
}

We adopt the notation from \citep{Vaart1998AsymptoticStatistics} that 
for a function $f(z; \theta)$ and a random vector $Z\sim P$, we write $Pf(\theta) = \E(f(Z;\theta))$ with $Z\sim P$.
Note that if $P_n$ is the empirical distribution of $Z_1,\cdots, Z_n, $
then $P_n f(\theta) = \frac{1}{n}\sum_{i=1}^n f(Z_i;\theta)$.

With the above notations, $\widehat \theta_n$ and $\theta_0$ are defined by the roots of $\Psi_n(\theta) = P_n\psi(\theta) =0$ and $\Psi(\theta) = P\psi(\theta) =0$ respectively.  In the logit-normal model, these roots are unique. 

Next denote
\al{
    \dot \Psi(\theta) = P \frac{d \psi(\theta)}{d\theta} = \int r(\gamma| Y)\frac{-d^2 \log(p(\gamma|x,\theta))}{d\theta^2} d\gamma
}
which is positive definite for all $\theta$, therefore $V = \dot \Psi(\theta_0)$ is non-singular. 

Since the vector of matrices $\ddot \psi(\theta) $ is continuous as a function of $\theta$, then in any open neighbourhood $\Theta_0$ of $\theta_0$ we can bound the magnitude all entries of $\ddot \psi(\theta) $ by a function $g$. 
Hence by the Taylor series expansion. 

\al{
    0 &= \Psi_n(\widehat \theta) = \Psi_n(\theta_0) + V (\widehat \theta - \theta_0) + \frac{1}{2}(\widehat \theta - \theta_0)^T\ddot \Psi_n(\widetilde \theta_n)(\widehat \theta - \theta_0)
}
where $\tilde \theta_{n,i} \in [\widehat \theta_i, \theta_{0,i}]$ or $[\theta_{0,i},\widehat \theta_i]$ by the mean value theorem.  Furthermore: 
\al{
    \norm{\ddot \Psi_n(\widetilde \theta_n)}_2 &= \norm{\frac{1}{n}\sum_{i = 1}^n \ddot \psi(\widetilde \theta_n) (X_i, Y_i)}_2 \\
    &\leq \frac{1}{n}\sum_{i = 1}^n \norm{\ddot \psi(\widetilde \theta_n) (X_i, Y_i)}_2 
}
Since $\ddot \psi(\widetilde \theta_n) (X_i, Y_i)$ is a 3D tensor constructed from taking three  derivatives of $\ell_1(\theta)(x,y)$ which will only depend on a sum terms consisting of a continuous function of a parameter, $\int r(\gamma|Y)d\gamma$, and one of $X$ or $XX^T$. Hence by the consistency of $\tilde \theta_n$ we can upper bound the functions of the parameters, and by the assumptions of Theorem \ref{supp_thm:par_closed_form}, then the remaining terms can be bounded by a random variable with finite mean $g(X,Y)$.  Then by the weak LLN $ \norm{\ddot \Psi_n(\widetilde \theta_n)} = O_p(n^{-1/2})$. 

Finally we have the result 
\begin{align*}
    \widehat \theta - \theta_0 &= - (V^{-1})\frac{1}{n}\sum_{i = 1}^n \psi(\theta_0)(X_i,Y_i) + o_p(n^{-1/2})
\end{align*}

and therefore 
\al{
    \sqrt{n}(\widehat \theta - \theta_0) &\stackrel{d}{\to} \mathcal{N}(0, V^{-1} P(\psi(\theta_0)^2)V^{-1})
}
\end{proof}

\section{Proofs of Main Paper Theorems}

\begin{theorem} \label{thm:linear_convergence} (Listed as Theorem 2.4 in the main paper.)
    Under the following assumptions, the NPEM algorithm will converge linearly (in the 1-norm) with rate $\rho \in [0,1)$ 
    
    Let $\mathcal{B}_\delta \subset [0,1]$ be a region of Lebesgue measure $\delta$.  Let $g^{(0)}(y) = p^{(0)}_{MA}(y) - \hat p(y)$ represent the difference of an initial estimate $p_{MA}^{(0)}$ and  $\hat p$. Suppose the following assumptions hold for some $\delta > 0$
    \begin{enumerate}
        \item[(L1)] $\norm{g^{(0)}}_1 \leq \epsilon c_0$ \\
        \item[(L2)] $\eta_1 \equiv \max_{y \not = y'}\sup_{\gamma \in \mathcal{B}_\delta} p_A(y|\gamma)p_A(y'|\gamma) \leq \frac{\rho c_0}{4 (N + 1)}$ \\
        \item[(L3)] $\eta_2 \equiv \max_{y} \sup_{\gamma \in \mathcal{B}_\delta} \sum_{y' \not = y} p_A(y|\gamma)p_A(y'|\gamma)\leq \frac{\rho c_0}{4 }$ \\
        \item[(L4)] $\inf_{y} \hat p(y)  = c_0 > 0 $ \\
        \item[(L5)] $\epsilon \leq \frac{\rho }{16}$
        \item[(L6)] $0 < \delta \leq \frac{\rho c_0 }{4 \overline p_M^{(0)}\exp(\frac{\rho}{1 - \rho}\frac{1}{c_0}\norm{ p_{MA}^{(0)}(\cdot) - \hat p(\cdot)}_1) (N + 2)}$,
        where $\overline p_M^{(0)} = \sup_{\gamma \in [0,1]} p_M^{(0)}(\gamma) $. 
    \end{enumerate}


Then the NPEM algorithm converges linearly at a rate of $\rho$ 
    \begin{equation*}
        \norm{p_{MA}^{(t + 1)}(y) - \hat p(y)}_1 \leq \rho \norm{p_{MA}^{(t)}(y) - \hat p(y)}_1.
    \end{equation*}
   
\end{theorem}

Before we prove Theorem \ref{thm:linear_convergence}, we 
first introduce a useful lemma.

\begin{lemma} \label{lem:oracle_inequality}
    By the assumptions (L1-6) of theorem \ref{thm:linear_convergence} the following inequality holds
    \al{
        \norm{ p^{(t + 1)}_{MA} - \hat p}_1 &\leq \Bigg( \frac{(N + 1) (\eta_1 + \overline p_M^{(t)} \delta)}{c_0} + \frac{\eta_2}{c_0} +\frac{\overline p_M^{(t)} \delta}{c_0} \\
        &+ \left(1 + \frac{1}{1 - \epsilon} \right)\frac{ \norm{ p^{(t)}_{MA} - \hat p}_1}{c_0} +  \frac{1}{1 - \epsilon}\frac{\norm{ p^{(t)}_{MA} - \hat p}_1^2}{c_0^2} \Bigg)\norm{ p^{(t)}_{MA} - \hat p}_1.
    }
   
\end{lemma}

\begin{proof}

Consider the error after update: 
\begin{equation}
    \begin{aligned}
    g^{(t + 1)}(y) &= p_{MA}^{(t + 1)}(y) - \hat p(y)\\
    &= \sum_{y'} \frac{\int p_A(y|\gamma)p_A(y'|\gamma)p_{M}^{(t)}(\gamma)d\gamma \hat p(y')}{p_{MA}^{(t)}(y')} - \hat p(y)
    \label{eq::L::1}
    \end{aligned}
\end{equation}

For now, assume that $g^{(0)}(y) \leq \epsilon c_0$ with $\epsilon < 1$, however, by the final results of the theorem, this will hold by induction for $g^{(t)}(y)$ and $t \geq 1$.  By the a Taylor expansion (geometric series),

\al{
    \frac{1}{ p_{MA}^{(t )}(y)}&=
    \frac{1}{\hat p(y') +g^{(t)}(y')} \\
    &= \frac{1}{\hat p(y')} - \frac{g^{(t)}(y')}{\hat p^2(y')}+ \frac{1}{\hat p(y')} \sum_{k = 2}^\infty \bigg( \frac{-g^{(t)}(y')}{\hat p(y')}\bigg)^k \\
    &= \frac{1}{\hat p(y')} - \frac{g^{(t)}(y')}{\hat p^2(y')}+ \frac{1}{\hat p(y')}\Bigg( \frac{\Big(\frac{g^{(t)}(y')}{\hat p(y')}\Big)^2}{1 + \frac{g^{(t)}(y')}{\hat p(y')}} \Bigg).
}
Putting the above result in to $p^{(t )}(y)$
in the denominator of equation \eqref{eq::L::1},
we obtain
\al{
    g^{(t + 1)}(y) &= \sum_{y'} \Bigg( \frac{\int p_A(y|\gamma)p_A(y'|\gamma)p_{MA}^{(t)}(\gamma)d\gamma \hat p(y')}{\hat p(y')}  - \frac{g^{(t)}(y')}{\hat p(y')}\int p_A(y|\gamma)p_A(y'|\gamma)p_{M}^{(t)}(\gamma)d\gamma \\
    &+ \int p_A(y|\gamma)p_A(y'|\gamma)p_{M}^{(t)}(\gamma)d\gamma \Bigg( \frac{\Big(\frac{g^{(t)}(y')}{\hat p(y')}\Big)^2}{1 + \frac{g^{(t)}(y')}{\hat p(y')}} \Bigg) \Bigg) - \hat p(y) \\
    &= p_{M}^{(t)}(y') - \hat p(y) -  \sum_{y'} \frac{g^{(t)}(y')}{\hat p(y')}\int p_A(y|\gamma)p_A(y'|\gamma)p_{M}^{(t)}(\gamma)d\gamma \\
    &+ \sum_{y'} \int p_A(y|\gamma)p_A(y'|\gamma)p_{M}^{(t)}(\gamma)d\gamma \Bigg( \frac{\Big(\frac{g^{(t)}(y')}{\hat p(y')}\Big)^2}{1 + \frac{g^{(t)}(y')}{\hat p(y')}} \Bigg) \\
    &= \underbrace{g^{(t)}(y) -  \sum_{y'}\frac{g^{(t)}(y')}{\hat p(y')}\int p_A(y|\gamma)p_A(y'|\gamma)p_{M}^{(t)}(\gamma)d\gamma}_{(A)} \\
    &+ \underbrace{\sum_{y'} \int p_A(y|\gamma)p_A(y'|\gamma)p_{M}^{(t)}(\gamma)d\gamma \Bigg( \frac{\Big(\frac{g^{(t)}(y')}{\hat p(y')}\Big)^2}{1 + \frac{g^{(t)}(y')}{\hat p(y')}} \Bigg) }_{(B)} 
}

Next we show that $(B)$ only contributes to higher order terms. 
\al{
    &\sum_{y'} \bigg| \int p_A(y|\gamma)p_A(y'|\gamma)p_{M}^{(t)}(\gamma)d\gamma \Bigg( \frac{\Big(\frac{g^{(t)}(y')}{\hat p(y')}\Big)^2}{1 + \frac{g^{(t)}(y')}{\hat p(y')}} \Bigg)  \bigg| \\
    &\leq  \left( \frac{1}{1 - \epsilon }\right) \sum_{y'} \bigg| \int p_A(y|\gamma)p_A(y'|\gamma)p_{M}^{(t)}(\gamma)d\gamma \Big(\frac{g^{(t)}(y')}{\hat p(y')}\Big)^2 \bigg|\\
    &\leq  \left( \frac{1}{1 - \epsilon }\right) \sum_{y'} \bigg| \int p_A(y'|\gamma)p_{M}^{(t)}(\gamma)d\gamma \Big(\frac{g^{(t)}(y')}{\hat p(y')}\Big)^2 \bigg|\\
    &\leq \left( \frac{1}{1 - \epsilon }\right) \sum_{y'} \bigg| (\hat p(y') + g^{(t)}(y')) \Big(\frac{g^{(t)}(y')}{\hat p(y')}\Big)^2 \bigg| \\
    &\leq \left( \frac{1}{1 - \epsilon }\right)\bigg( \frac{\norm{g^{(t)}(\cdot)}_1^2}{c_0} + \frac{\norm{g^{(t)}(\cdot)}_1^3}{c_0^2} \bigg).
}
Hence, the term (B) will contribute only to super-linear terms. 

Returning to the term (A). Let  
\al{
    (T^*) &:= \sum_{y}\bigg| g^{(t)}(y) -  \sum_{y'} \frac{g^{(t)}(y')\int p_A(y|\gamma)p_A(y'|\gamma)p_{M}^{(t)}(\gamma)d\gamma }{\hat p(y')}\bigg|
}
and it is easy to see that $|(A)|\leq (T^*)$.

We can further bound $(T^*)$ by      
\al{
    (T^*) &\leq \underbrace{\sum_{y}\bigg| g^{(t)}(y) -  \frac{g^{(t)}(y)\int p_A(y|\gamma)p_A(y|\gamma)p_{M}^{(t)}(\gamma)d\gamma }{\hat p(y)}\bigg|}_{(T_1)} \\
    &+ \underbrace{\sum_{y}\bigg| \sum_{y' \not = y} \frac{g^{(t)}(y')\int p_A(y|\gamma)p_A(y'|\gamma)p_{M}^{(t)}(\gamma)d\gamma }{\hat p(y')} \bigg|}_{(T_2)}.
}

{\bf Bounding $T_1$:}
\al{
    T_1 &= \sum_{y}\bigg| g^{(t)}(y) -  \frac{g^{(t)}(y)\Big(\int_{\mathcal{B}_\delta} p_A(y|\gamma)p_A(y|\gamma)p_{M}^{(t)}(\gamma)d\gamma + \int_{\mathcal{B}_\delta^c} p_A(y|\gamma)p_A(y|\gamma)p_{M}^{(t)}(\gamma)d\gamma\Big)}{\hat p(y)}\bigg| \\
    &\leq \sum_{y}\bigg| g^{(t)}(y) -  \frac{g^{(t)}(y)\int_{\mathcal{B}_\delta} p_A(y|\gamma)p_A(y|\gamma)p_{M}^{(t)}(\gamma)d\gamma}{\hat p(y)}\bigg| + \frac{\overline p^{(t)} \delta \norm{g^{(t)}( \cdot)}_1}{c_0} }.

Let $(**) = 1 - \frac{\int_{\mathcal{B}_\delta} p_A(y|\gamma)p_A(y|\gamma)p_{M}^{(t)}(\gamma)d\gamma}{\hat p(y)}$.

Using the sharpness assumption (L2),
\begin{equation}
    p_A(y|\gamma)p_A(y|\gamma) = p_A(y|\gamma)(1 - \sum_{y'}p_A(y'|\gamma)) \geq p_A(y|\gamma) - \eta_2. 
    \label{eq:overlap_assumption}
\end{equation}

If $(**) \geq 0$, by equation \eqref{eq:overlap_assumption}
\al{
    \bigg| g^{(t)}(y) -  \frac{g^{(t)}(y)\int_{\mathcal{B}_\delta} p_A(y|\gamma)p_A(y|\gamma)p_{M}^{(t)}(\gamma)d\gamma}{\hat p(y)}\bigg| &\leq \bigg| g^{(t)}(y) -  \frac{g^{(t)}(y)\int_{\mathcal{B}_\delta} (p_A(y|\gamma) - \eta_2)p_{M}^{(t)}(\gamma)d\gamma}{\hat p(y)}\bigg| 
}
otherwise if $(**) < 0$,
\al{
    \bigg| g^{(t)}(y) -  \frac{g^{(t)}(y)\int_{\mathcal{B}_\delta} p_A(y|\gamma)p_A(y|\gamma)p_{M}^{(t)}(\gamma)d\gamma}{\hat p(y)}\bigg| &\leq \bigg| g^{(t)}(y) -  \frac{g^{(t)}(y)\int_{\mathcal{B}_\delta} p_A(y|\gamma)p_{M}^{(t)}(\gamma)d\gamma}{\hat p(y)}\bigg| .
}
As a result,
\al{
    T_1 &\leq  \sum_{y}\bigg| g^{(t)}(y) -  \frac{g^{(t)}(y)\int_{\mathcal{B}_\delta} p_A(y|\gamma)p_{M}^{(t)}(\gamma)d\gamma}{\hat p(y)}\bigg| + \frac{\eta_2 \norm{g^{(t)}(\cdot)}_1}{c_0} +\frac{\overline p^{(t)} \delta \norm{g^{(t)}(\cdot)}_1}{c_0} \\
    &= T_3 + \frac{\eta_2 \norm{g^{(t)}(\cdot)}_1}{c_0} +\frac{\overline p^{(t)} \delta \norm{g^{(t)}(\cdot)}_1}{c_0},
}
where  
\al{
    T_3 &:= \sum_{y}\bigg| g^{(t)}(y) -  \frac{g^{(t)}(y)\int_{\mathcal{B}_\delta} p_A(y|\gamma)p_{M}^{(t)}(\gamma)d\gamma}{\hat p(y)}\bigg| \\
    &= \sum_{y}\bigg| g^{(t)}(y) -  \frac{g^{(t)}(y)(\int p_A(y|\gamma)p_{M}^{(t)}(\gamma)d\gamma- \int_{\mathcal{B}_\delta^c} p_A(y|\gamma)p_{M}^{(t)}(\gamma)d\gamma) }{\hat p(y)}\bigg| \\
    &\leq \sum_{y}\bigg| g^{(t)}(y) -  \frac{g^{(t)}(y)(\hat p(y) + g^{(t)}(y))}{\hat p(y)}\bigg| +\frac{\overline p^{(t)} \delta \norm{g^{(t)}(\cdot)}_1}{c_0} \\
    &\leq \sum_{y}\frac{| g^{(t)}(y)|^2}{c_0} +\frac{\overline p^{(t)} \delta \norm{g^{(t)}(\cdot)}_1}{c_0} \\
    &\leq \frac{\norm{g^{(t)}(\cdot)}_1^2}{c_0} +\frac{\overline p^{(t)} \delta \norm{g^{(t)}(\cdot)}_1}{c_0}
}

{\bf Bounding $T_2$:}
\al{
    T_2 &= \sum_{y}\bigg| \sum_{y' \not = y} \frac{g^{(t)}(y')\int p_A(y|\gamma)p_A(y'|\gamma)p_{M}^{(t)}(\gamma)d\gamma }{\hat p(y')} \bigg| \\
    &\leq \sum_{y} \sum_{y' \not = y} \frac{|g^{(t)}(y')|\Big(\int_{\mathcal{B}_\delta} p_A(y|\gamma)p_A(y'|\gamma)p_{M}^{(t)}(\gamma)d\gamma +  \int_{\mathcal{B}_\delta^c} p_A(y|\gamma)p_A(y'|\gamma)p_{M}^{(t)}(\gamma)d\gamma\Big)}{c_0} \\
    &\leq \sum_{y} \sum_{y' \not = y} \frac{|g^{(t)}(y')|\Big(\eta_1 + \overline p^{(t)} \delta\Big)}{c_0} \\
    &\leq  \frac{(N + 1) (\eta_1 + \overline p^{(t)} \delta)\norm{g^{(t)}(\cdot)}_1}{c_0}.
}

Hence, combining all of these terms, we conclude that
\al{
    \norm{ p^{(t + 1)}(\cdot) - \hat p(\cdot)}_1 &\leq \Big( \frac{(N + 1) (\eta_1 + \overline p^{(t)} \delta)}{c_0} + \frac{\eta_2}{c_0} +\frac{\overline p^{(t)} \delta }{c_0} + \left(1 + \frac{1}{1 - \epsilon} \right) \frac{ \norm{ p_{MA}^{(t)}(\cdot) - \hat p(\cdot)}_1}{c_0} \\
    &+ \left(\frac{1}{1 - \epsilon} \right) \frac{\norm{ p_{MA}^{(t)}(\cdot) - \hat p(\cdot)}_1^2}{c_0^2} \Big)\norm{ p_{MA}^{(t)}(\cdot) - \hat p(\cdot)}_1,
}
which completes the proof.

\end{proof}

We now continue with the remainder of the proof for theorem \ref{thm:linear_convergence}
    
\begin{proof}[Proof of Theorem \ref{thm:linear_convergence}]

By Lemma~\ref{lem:oracle_inequality},
\begin{align*}
\norm{ p^{(t + 1)}_{MA} - \hat p}_1 &\leq \Bigg( \frac{(N + 1) (\eta_1 + \overline p_M^{(t)} \delta)}{c_0} + \frac{\eta_2}{c_0} +\frac{\overline p_M^{(t)} \delta}{c_0} \\
        &+ \left(1 + \frac{1}{1 - \epsilon} \right)\frac{ \norm{ p^{(t)}_{MA} - \hat p}_1}{c_0} +  \frac{1}{1 - \epsilon}\frac{\norm{ p^{(t)}_{MA} - \hat p}_1^2}{c_0^2} \Bigg)\norm{ p^{(t)}_{MA} - \hat p}_1\\
        & = \Bigg( \underbrace{\frac{(N + 1) \eta_1 +\eta_2 }{c_0}}_{(I)} + \underbrace{\frac{(N+2)}{c_0}\overline p_M^{(t)} \delta}_{(II)} \\
        &+ \underbrace{\left(1 + \frac{1}{1 - \epsilon} \right)\frac{ \norm{ p^{(t)}_{MA} - \hat p}_1}{c_0} +  \frac{1}{1 - \epsilon}\frac{\norm{ p^{(t)}_{MA} - \hat p}_1^2}{c_0^2}}_{(III)} \Bigg)\norm{ p^{(t)}_{MA} - \hat p}_1
\end{align*}
Thus, all we need is to bound $(I)+(II)+(III)\leq \rho$.

We use the proof by induction.
First, we prove the case of $t=0$.

{\bf Case $t=0$ -- Bounding (I):}
By assumptions (L1-2), it is clear that $(I)\leq \frac{\rho}{2}$.
So we focus on other two terms.

{\bf Case $t=0$ -- Bounding (II):}
Assumption (L6) implies that 
$$
\overline p_M^{(t)} \delta\leq \frac{\rho c_0}{4\exp(\frac{\rho}{1-\rho} \frac{2}{c_0} \|g^{(0)}\|_1) N (N+2)}\leq \frac{\rho c_0}{4 (N+2)}.
$$
Thus, 
$$
(II)\leq \frac{\rho}{4}.
$$

{\bf Case $t=0$ -- Bounding (III):}
By assumptions (L4-5),
\al{
    \left( 1 + \frac{1}{1 - \epsilon} \right) \frac{ \norm{ p_{MA}^{(0)}(\cdot) - \hat p(\cdot)}_1}{c_0} + \left(\frac{1}{1 - \epsilon} \right)\frac{\norm{ p_{MA}^{(0)}(\cdot) - \hat p(\cdot)}_1^2}{c_0^2} &\leq  \frac{3}{1 - \epsilon} \frac{ \norm{ p_{MA}^{(0)}(\cdot) - \hat p(\cdot)}_1}{c_0} \\
    &\leq \frac{\rho}{4}.
}
Putting all the above three bounds together,
we conclude that 
$(I)+(II)+(III)\leq \rho$, which proves the case $t=0$.

{\bf Case $t=1,2,3,\cdots.$}
Now we assume that the conclusion holds for iterations $0,1,2,3,\cdots, t$. 
And we will show that it also holds for $t+1$.
Thus, we have
$$
\norm{ p^{(j)}_{MA} - \hat p}_1\leq \rho \norm{ p^{(j-1)}_{MA} - \hat p}_1
$$
for $j=1,2,3,\cdots, t$.
Since the bound on (I) also applies to $t=1,2,3,\cdots,$
we can directly use the result in $t=0$. So $(I)\leq \rho/2$.

{\bf Case $t=1,2,3, \cdots$ -- Bounding (II):}

For the case $j= t+1:$
\al{
    p^{(t + 1)}(\gamma) &= p_{M}^{(t)}(\gamma) \sum_y \frac{\hat p(y)p_A(y|\gamma)}{p_{MA}^{(t)}(y)} \\
    &\leq p_{M}^{(t)}(\gamma) \sum_y \frac{\hat p(y)}{p_{MA}^{(t)}(y)} \\
    &= p_{M}^{(t)}(\gamma) \sum_y \bigg(1 - \frac{\frac{g^{(t)}(y)}{\hat p(y)}}{1 + \frac{g^{(t)}(y)}{\hat p(y)}}\bigg) \\
    &\leq p_{M}^{(t)}(\gamma)\bigg(1 + \left(\frac{1}{1 - \epsilon}\right)\frac{\norm{g^{(t)}(\cdot)}_{1}}{c_0}\bigg).
}

Thus, the upper bound $\overline p_{M}^{(t+1)}\leq \overline p_{M}^{(t)}(1 + \left(\frac{1}{1 - \epsilon}\right)\frac{1}{c_0}\rho^t \norm{g^{(0)}(\cdot)}_{1})$. Since $\log(1 + x) \leq x$ for $x > -1$,

\al{
    \overline p_{M}^{(t+1)} &\leq \overline p_{M}^{(0)}\prod_{k = 1}^t (1 + \left(\frac{1}{1 - \epsilon}\right)\frac{1}{c_0}\rho^k \norm{g^{(0)}(\cdot)}_{1}) \\
    \log(\overline p_{M}^{(t+1)}) &\leq \log( \overline p_{M}^{(0)}) + \left(\frac{1}{1 - \epsilon}\right)\frac{1}{c_0} \norm{g^{(0)}(\cdot)}_{1}\sum_{k = 1}^t\rho^k \\
    &\leq  \log(\overline p_{M}^{(0)}) + \left(\frac{1}{1 - \epsilon}\right)\frac{1}{c_0} \norm{g^{(0)}(\cdot)}_{1} \frac{\rho}{1 - \rho}.
}
Hence, $\overline p_{M}^{(t+1)} \leq \overline p_{M}^{(0)} \exp\bigg( \left(\frac{1}{1 - \epsilon}\right)\frac{1}{c_0} \norm{g^{(0)}(\cdot)}_{1} \frac{\rho}{1 - \rho}\bigg).$

Now returning to the term $(II)$. A direct application of the above inequality leads to \begin{align*}
    (II) &= \frac{(N+2)}{c_0}\overline p_M^{(0)} \delta\\
    &\leq \frac{(N+2)}{c_0}\exp\bigg( \left(\frac{1}{1 - \epsilon}\right)\frac{1}{c_0} \norm{g^{(0)}(\cdot)}_{1} \frac{\rho}{1 - \rho}\bigg) \overline p_M^{(0)} \delta \\
    &\overset{(L6)}{\leq} \frac{\rho c_0}{4 (N+2)}\\
    &\leq \frac{\rho}{4}.
\end{align*}

{\bf Case $t=1,2,3, \cdots$ -- Bounding (III):}

Using the induction and the fact that $\rho<1$,
\begin{align*}
(III)& = 
\left( 1 + \frac{1}{1 - \epsilon} \right) \frac{ \norm{ p_{MA}^{(t)}(\cdot) - \hat p(\cdot)}_1}{c_0} + \left(\frac{1}{1 - \epsilon} \right)\frac{\norm{ p_{MA}^{(t)}(\cdot) - \hat p(\cdot)}_1^2}{c_0^2} \\
&\leq  \frac{3}{1 - \epsilon} \frac{ \norm{ p_{MA}^{(t)}(\cdot) - \hat p(\cdot)}_1}{c_0} \\
    &\leq  \frac{3}{1 - \epsilon} \frac{ \rho^t\norm{ p_{MA}^{(0)}(\cdot) - \hat p(\cdot)}_1}{c_0} \\
    &\leq  \frac{1}{4} \frac{ \rho\norm{ p_{MA}^{(0)}(\cdot) - \hat p(\cdot)}_1}{c_0} \\
    &\overset{(L4-5)}{\leq} \frac{\rho}{4}.
\end{align*}

As a result, for iteration $t+1$,
we sill have $(I)+(II)+(III)\leq \rho$.
By induction, this holds for all $t$
and the result follows.





\end{proof}

\begin{theorem} \label{thm:regularized_EM} (Listed as Theorem 2.3 in the main paper.)
    Denote the unique global solution $P^*_{M,\mu} = \argmax_{P_M \in \mathcal{P}_{[0,1]}} \ell_{n,\mu}[P_M]$.
    
    Then consider a sequence of latent trait distributions $\{P^{(t)}_{M,\mu}\}_{t = 0}^{\infty}$ generated by the EM algorithm for the regularized likelihood. If $\mu > 0$ and $p_A(y|\gamma)$ is continuous in $\gamma$ for each $y$, then 
    
    $$ \ell_{n,\mu}[P^{(t)}_{M,\mu}] \stackrel{t \to \infty}{\longrightarrow} \ell_{n,\mu}[P^{*}_{M,\mu}].$$ 
    
    and
    
    $$ P^{(t)}_{M,\mu} \stackrel{t \to \infty}{\longrightarrow_w} P^{*}_{M,\mu}$$
    where $\longrightarrow_w$ denotes weak convergence of measures. 
\end{theorem}


Before proving the main theorem, we prove a series of lemmas which will be useful in characterizing a maximizer. 
 
\begin{lemma} \label{lem:continuous_maximizer}
    If $\mu >0$ and $p_A(y|\gamma)$ is continuous in $\gamma$ for each $y$ then the maximizer of $\ell_{n,\mu}[\cdot]$, $P^*_M$ and $P_U$ are mutually absolutely continuous i.e. $P_M^* \ll P_U$ and $P_U \ll P_M^*$
\end{lemma}
\begin{proof}
    Firstly, if $P_U \not \ll P_M^*$ then $\KL(P_U ||P_M^*) = \infty$ and clearly $P_M^*$ is not a maximizer. Next consider the Lebesgue decomposition theorem \citep{Halmos2013MeasureTheory}. Consider an arbitrary probability measure $Q$ which we compare to $P_U$ (which on the interval is equivalent to the Lebesgue measure).  Unless otherwise stated, singular and continuous will be with respect to this measure $P_U$.  There exist measures $Q_c$ and $Q_s$ such that: 
    \al{
        Q &= Q_c + Q_s \\
        \text{Where } Q_c &\ll P_U \\
        Q_s &\perp P_U
    }
    I.e. $Q_c$ is absolutely continuous with respect to $P_U$ and $Q_s$ is singular with respect to $P_U$. \\
    Next, we consider the KL divergence between $P_U$ and $Q$. Since $Q$ is a probability measure, we can rescale the measures and define $Q = (1 - \lambda)Q_c + \lambda Q_s$ where $Q_c$ and $Q_s$ are propability measures as well. Then
    \al{
        \KL(P_U || Q) &= \int \log\left( \frac{dP_U}{dQ}\right) dP_U \\
        &= \int \log\left( \frac{dP_U}{ (1 - \lambda)dQ_c + \lambda dQ_s}\right) dP_U \\
    }
    Denote $\mathcal{X}_s \subset [0,1]$ the subset of the interval where the singular part of the measure is defined. Since $P_U(\mathcal{X}_s) = 0$ on this subset, $\frac{dP_U}{ (1 - \lambda)dQ_c + \lambda dQ_s}(x) = 0$ for all $x \subset \mathcal{X}_s$.  Hence the KL divergence can be computed directly using only the absolutely continuous part, though this will be scaled according to $\lambda$. 
    \al{
        \KL(P_U || Q) &= \int \log\left( \frac{dP_U}{ (1 - \lambda)dQ_c + \lambda dQ_s}\right) dP_U \\
        &= \int_{[0,1]} \log\left( \frac{dP_U}{ (1 - \lambda)dQ_c + \lambda dQ_s}(x)\right) dP_U(x) \\
        &= \int_{[0,1]} \log\left( \frac{dP_U}{ (1 - \lambda)dQ_c}(x)\right) dP_U(x) \\
        &= -\log(1 - \lambda)   + \int_{[0,1]} \log\left( \frac{dP_U}{dQ_c}(x)\right) dP_U(x)
    }
    
    Secondly we verify that that $\ell_{n,\mu}[P_M]$ is a concave function of $P_M$. We can add a constant to $\ell_{n,\mu}[P_M]$ so that
    \al{
        \ell_{n,\mu}[P_M] - \sum_{y}\hat p(y)\log(\hat p(y)) &= -\mathcal{D}(\hat p|| p_{MA}) - \mu \mathcal{D}(p_U|| p_{M}). 
    }
    Since $p_{MA}$ is linear in $p_M$ and the KL divergence is a strictly concave function, then this regularized likelihood is clearly a strictly concave function in $p_M$.  
    
    Next we consider Theorem 3.4.3 in \citep{Kosmol2011OptimizationSpaces} we have a sufficient and necessary condition for the optimality of a measure $P^*_M$. 
    
    If $K$ is a convex subset of a vector space $X$ and $f : X \to \mathbb{R}$ a convex function, $f$ has a minimal solution $x_0 \in K$ if and only if for all $x \in K$
    \begin{equation}
        \dot f_+(x_0;x - x_0 ) \geq 0
    \end{equation}
    Where 
    \al{
        \dot f_+(x_0;z) &= \lim_{t \to 0^+ } \frac{f(x_0 + tz) - f(x_0)}{t} 
    }
    
    In our case, we define $X$ to be the vector space of signed measures and $K$ probability measures over $[0,1]$.  We can also use this to show if $x$ is a direction which adds a singular measure to a particular $\tilde x_0$ then the Gateaux derivative will be negative, and hence any measure which increases the likelihood cannot have any singular component. Since our log-likelihood functional $\ell_{n,\mu}[\cdot]$ is concave, the theorem above will apply with a flip of the inequality and defining $f = -\ell_{n,\mu}$. \\
    Denote, for any measure $H$, $p_A(y;H) = \int p_A(y|\gamma)dH(\gamma)$. 
    
    Consider $Q$ to be an continuous measure, and define a direction $G = Q_s - Q$ where $Q_s$ is singular. Let $q$ be the corresponding density function to $Q$. 
    \al{
        \ell_{n,\mu}[Q + \lambda G] &= \sum_{y}\hat p(y) \log\left( p_A(y;Q) + \lambda p_A(y;G)\right) - \mu \int_0^1 \log(q(\gamma)) d\gamma \\
        &+\mu\log(1 - \lambda) \\
        &=  \sum_{y}\hat p(y)\left( \log\left( p_A(y;Q)\right) + \log\left( 1 + \lambda \frac{p_A(y;G)}{p_A(y;Q)}\right)\right) - \mu \int_0^1 \log(q(\gamma)) d\gamma \\
        &+\mu\log(1 - \lambda) \\
        &= \ell_{n,\mu}[Q] +  \sum_{y}\hat p(y)\log\left( 1 + \lambda \frac{p_A(y;G)}{p_A(y;Q)}\right) + \mu\log(1 - \lambda) \\
        &=  \ell_{n,\mu}[Q] + \lambda \sum_{y}\hat p(y)\frac{p_A(y;G)}{p_A(y;Q)} + \mu\log(1 - \lambda) + O(\lambda^2) \\
        \implies \dot \ell_{n,\mu}(Q;G) &=  \sum_{y}\hat p(y)\frac{p_A(y;Q_s)}{p_A(y;Q)} - 1 - \mu
    }
    We note that by \cite{Lindsay1995MixtureApplications} in the unregularized case, 
    $Q^*$ is a maximizer of $\ell_{n, \mu = 0}[\cdot]$ if and only if
    $$ \sum_{y}\frac{\hat p(y)p_A(y|\gamma)}{p_{A}(y;Q^{*})}  - 1 \leq 0  $$
    for all $\gamma$.  
    Clearly integrating this over any measure $Q'$ it will also hold that: 
    $$ \sum_{y}\frac{\hat p(y)p_{A}(y;Q')}{p_{A}(y;Q^{*})}  - 1 \leq 0 $$
    Hence we consider 2 cases, whether or not if there exists a continuous measure $Q^*$ which maximizes the unregularized problem. 
    
    \textbf{Case 1: } \textit{There exists $Q^*$ which is continuous which maximizes the unregularized likelihood ($\ell_{n,\mu = 0}$).} 
    
    If this holds then 
    \al{
        \sum_{y}\frac{\hat p(y)p_A(y|\gamma)}{p_{A}(y;Q^{*})}  - 1  &< \sum_{y}\frac{\hat p(y)p_A(y|\gamma)}{p_{A}(y;Q^{*})}  - 1 - \mu \leq -\mu < 0.
    }
    And hence any direction which adds a singular measure will decrease the likelihood. 
    
    \textbf{Case 2: } \textit{Any maximizer $Q^*$ of the unregularized likelihood ($\ell_{n,\mu = 0}$) is not continuous.} \\
    
    If the only maximizers of the unregularized problem are not continuous, then by the fundamental theorem of mixture models there exists a maximizer with at most $N + 1$ discrete support points  \citep{Lindsay1995MixtureApplications}. If $p_A(y|\gamma)$ are all continuous in $\gamma$ for each $y$ then we can define a measure $\tilde Q_\delta$ such that $\tilde Q_\delta$ places the same point mass uniformly within a region of $\pm \delta$ of each of the point masses (some of which may only be $+$ or $-$ if the mass is near the boundary $0$ or $1$) then for all $\delta > 0, \tilde Q_\delta$ is continuous measure.  Since $p_A(y|\gamma)$ is continuous, $p_A(y; Q^*) = p_A(y; \tilde Q_\delta) + \varepsilon(y;\delta)$ and for all $\epsilon > 0 $ there exists a $\delta > 0 $ such that $ |\varepsilon(y;\delta)| < \epsilon $ for all $y$. 
    
    Therefore, computing the derivative in the direction of any singular measure from $\tilde Q_\delta$
    \al{
        \dot \ell_{n,\mu}(\tilde Q_\delta;G) &= \sum_{y}\hat p(y)\frac{p_A(y;Q_s) }{p_A(y;\tilde Q_\delta)} - 1 - \mu \\
        &= \sum_{y}\hat p(y)\frac{p_A(y;Q_s) }{p_A(y;\tilde Q_\delta)} - 1 - \mu \\
        &= \sum_{y}\hat p(y)\frac{p_A(y;Q_s) }{p_A(y;Q^{*})}\left( \frac{1}{1 + \frac{\varepsilon(y; \delta)}{p_A(y; Q^*)}}\right)  - 1 - \mu
    }
    Of course, we can pick $\epsilon$ such that $\epsilon < \frac{1}{2}\inf_y p_A(y; Q^*)$ then
    \al{
        \left| \frac{1}{1 + \frac{\varepsilon(y; \delta)}{p_A(y; Q^*)}}\right| &\leq 1 + 2  \frac{\left|\varepsilon(y; \delta)\right|}{p_A(y; Q^*)}.
    }
    Now we also ensure $\epsilon \leq \min\{\mu, \frac{1}{2}\} \inf_y p_A(y; Q^*) $ and therefore: 
    \al{
        \left| \frac{1}{1 + \frac{\varepsilon(y; \delta)}{p_A(y; Q^*)}}\right| &\leq (1 + \mu)
    }
    Therefore since there exists $\delta > 0 $ such that the above holds, $\tilde Q_\delta$ is continuous and the Gateaux derivative in the direction of any singular measure is negative, therefore the maximizer must be a continuous measure $(P^*_M \ll P_U)$, and therefore we can express the maximizer in terms of its density function. 
\end{proof}

We next introduce another lemma, which will justify that the density function of the  maximizer will be bounded below
\begin{lemma}\label{lem:bounded_below}
    The maximizer of $\ell_{n, \mu}$,  $P_M^*$  satisfies  
    \begin{equation}
        \essinf p_M^*(\gamma) \geq \frac{\mu}{1 + \mu}
    \end{equation}
    where $p_M^*$ is the corresponding density function of the measure $P_M^*$, and $ \essinf f$ is the essential infimum of f, the maximum value $\alpha$ such that the set $\{x: f(x) < \alpha\}$ has (Lebesgue) measure $0$. 
\end{lemma}

\begin{proof}
    We prove this by contradiction and using the monotonicity of the EM algorithm.  Suppose $\tilde P_M$ is a maximizer such that  $\essinf_{\gamma} \tilde P_M(\gamma) < \frac{\mu}{1 + \mu}$. Denote $\tilde P_{M_2}(\gamma)$ as the density function from taking one regularized NPEM step from $\tilde P_{M}(\gamma)$. 
    \al{
        \ell_{n, \mu}[\tilde P_{M_2}] - \ell_{n, \mu}[\tilde P_{M}] &\geq (1 + \mu) \mathcal{D}(\tilde P_{M_2} || \tilde P_{M})
    }
    But clearly by the monotonicity of the EM algorithm.  
    \al{
        \tilde p_{M_2}(\gamma) &= \left( \frac{1}{1 + \mu}\right)\sum_{y}\frac{\hat p(y)p_A(y|\gamma)\tilde p_{M}(\gamma)}{\tilde p_{MA}(y)} + \left( \frac{\mu}{1 + \mu}\right) \\
        &\geq \frac{\mu}{1 + \mu}
    }
    And if $ \tilde p_{M_2} \not =  \tilde p_{M}$ except for on a region of Lebesgue measure $0$ then $\mathcal{D}(\tilde P_{M_2}||\tilde P_M) > 0$ and we guarantee $\ell_{n, \mu}[\tilde P_{M_2}] > \ell_{n, \mu}[\tilde P_{M}] $ and hence $\tilde P_{M}$ is not a maximizer, a contradiction.
\end{proof}

Next we introduce a final lemma before proving the main theorem. 
\begin{lemma} \label{lem:optimality_criterion}
    A distribution (measure) $P_M^*$ is a maximizer of $\ell_{n, \mu}$ if and only if: 
    \begin{equation}
         \sum_{y}\frac{\hat p(y)p_A(y|\gamma)}{p_{MA}^{*}(y)} + \frac{\mu}{p_M^*(\gamma)} - 1 - \mu \leq 0 \quad a.e.
    \end{equation}
    Where $p_M^*$ is the corresponding density function to $P_M^*$
\end{lemma}

\begin{proof}
     Recalling Theorem 3.4.3 in \citep{Kosmol2011OptimizationSpaces} we have a sufficient and necessary condition for the optimality of a measure $P^*_M$. 
     
     We can define a direction $G = \tilde P_M - P^*_M$ and by lemma \ref{lem:continuous_maximizer} we only need to consider the directions $G$ which admit a density function and therefore can express the directions in terms of the difference of densities $g$. Then  
    \al{
        \ell_{n,\mu}[P^*_M + tG] &= \sum_{y} \hat p(y)\log\left( \int p_A(y|\gamma)\left(p_M^*(\gamma) + tg(\gamma)\right)d\gamma\right) + \mu \int \log\left(p_M^*(\gamma) + tg(\gamma)\right)d\gamma \\
        &= \sum_{y} \hat p(y)\log\left( \int p_A(y|\gamma)p_M^*(\gamma)d\gamma\right) + \mu \int \log\left(p_M^*(\gamma)\right)d\gamma  + \\
        &\sum_{y} \hat p(y) \log\left(1 +  t \int  \frac{p_A(y|\gamma)g(\gamma)}{p^*_{MA}(y)}d\gamma\right) + \mu \int \log\left(1 + t\frac{g(\gamma)}{p^*_M(\gamma)}\right)d\gamma \\
        &= \ell_{n,\mu}[P^*_M] + t \int \left( \sum_{y} \frac{\hat p(y)}{p^*_{MA}(y)}p_A(y|\gamma) + \mu \frac{1}{p^*_M(\gamma)} \right)  g(\gamma) d\gamma + O(t^2) \\
        \implies \dot \ell_{n,\mu}(Q;G) &= \int \left( \sum_{y} \frac{\hat p(y)}{p^*_{MA}(y)}p_A(y|\gamma) + \mu \frac{1}{p^*_M(\gamma)} \right)  g(\gamma)d\gamma \\
        &= \sum_{y} \frac{\hat p(y)}{p^*_{MA}(y)}\int p_A(y|\gamma)p_M(\gamma)d\gamma + \mu \int \frac{\tilde p_M(\gamma)}{p^*_M(\gamma)}d\gamma - 1 - \mu 
    }
    Since this holds for any $\tilde P_M$ which are continuous, at each point on $\gamma \in [0,1]$ this will hold for a sequence of continuous measures converging weakly to a point mass for each point on the interval.  Since $P_M^*$ is a continuous measure, it has a corresponding density function which has discontinuities at most on a set of Lebesgue measure $0$. Therefore according to the portmanteau lemma
    \al{
         \sum_{y} \frac{\hat p(y)}{p^*_{MA}(y)} p_A(y|\gamma) + \mu \frac{1}{p^*_M(\gamma)} - 1 - \mu &\leq 0 \quad a.e. 
    }
    
    We note that with a minor abuse of notation, if $g \in L^2_{[0,1]}$ then we could express the Gateaux derivative using the inner product of a functional gradient element $\nabla  \ell_{n,\mu}[P_M]$ and the direction $g$.  
    
    $$\nabla  \ell_{n,\mu}[P_M] = \sum_{y} \frac{\hat p(y)}{p_{MA}(y)}p_A(y|\cdot) + \mu \frac{1}{p_M(\cdot)}$$ and 
    $$ \dot f_+(P_M;G) = \langle \nabla  \ell_{n,\mu}[P_M], g\rangle .$$
    
    This functional gradient itself is a function of $\gamma$ and is useful for defining the NPEM algorithm $$p_M^{(t + 1)}(\gamma) = \frac{\nabla  \ell_{n,\mu}[P_M](\gamma)}{1 + \mu}p_M^{(t )}(\gamma),$$ as well as the optimality criterion $$\nabla  \ell_{n,\mu}[P_M](\gamma) \leq 1 + \mu \quad a.e.$$
    
    \textbf{Remark: } As we proved in Lemma \ref{lem:bounded_below}, $1/p^*_M(\gamma)$ will be defined due to the lower bound.  We compare this result to the optimality criterion of \citep{Lindsay1995MixtureApplications} a similar criteria is developed for the unregularized mixing distribution estimation problem, for which $p^*_M$ is a maximizer if and only if
    $$ \sum_{y}\frac{\hat p(y)p_A(y|\gamma) p^*_M(\gamma)}{p_{MA}^{*}(y)}  - 1 \leq 0 \quad \forall \gamma$$
    which is nearly equal to our condition when $\mu = 0$.  The main difference is due to the possibility of discontinuity in $\frac{1}{p_M^*(\gamma)}$ lemma \ref{lem:optimality_criterion} only holds up to a set of Lebesgue measure $0$ while the unregularized case holds for all $\gamma$. 
\end{proof}

Now we have all the required lemmas to prove the main part of the theorem. 

\begin{proof}[Proof of Theorem \ref{thm:regularized_EM}]

We use a similar technique as in the convergence of the mixing distribution.  \citep{Chung2015ConvergenceDistributions}.  Denote the sequence of measures generated by the EM algorithm as $\{P^{(t)}_M\}_{t = 0}^\infty$.  Since these are all measures defined on a compact i.e. $[0,1] \subset \mathbb{R}$, this sequence is tight. By Prokorov's theorem \citep{Billingsley2013ConvergenceMeasures}, this sequence is also sequentially compact. 

Next, due to the fact this sequence is sequentially compact, there must exist a convergent sub-sequence $\{P^{(t_k)}_M\}_{k = 0}^\infty$, where with limit $P^{**}_M$. Note that since $\ell_{n, \mu}[P^{(t)}_M]$ is monotone, and $ \ell_{n, \mu}[P^{(t + 1)}_M] - \ell_{n, \mu}[P^{(t)}_M] \geq \mathcal{D}(P^{(t)}_M||P^{(t + 1)}_M)$ every sub-sequence must have their likelihood converge to the same likelihood value i.e. $\ell_{n, \mu}[P^{**}_M]$.  Suppose that $P^{**}_M$ is not the global optimum for $\ell_{n, \mu}$, then by Lemma \ref{lem:optimality_criterion} there must be some region $ G^* \subset [0,1]$  with non-zero Lebesgue measure where $\nabla \ell_{n, \mu}[P_M^{**}](\gamma) \geq \delta' > 1 + \mu $ for $\gamma\in G^*$.  Since $\{P^{(t_k)}_M\}_{k = 0}^\infty$ converges weakly to  $P_M^{**}$ then $\nabla \ell_{n, \mu}[P_M^{(t_k)}](\gamma) \geq \delta > 1 + \mu $ for $\gamma\in G^*$ for all $k$ greater than some $K$.

Thus, 
\al{
    p^{(t_{k + 1})}_M(\gamma) &= \left(\frac{\nabla \ell_{n, \mu}[P_M](\gamma)}{1 + \mu}\right)^{t_{k + 1} - t_{k}}p^{(t_{k + 1})}_M(\gamma) \geq \frac{\delta}{1 + \mu} p^{(j )}_M(\gamma)
}

which further implies since ${t_{k + 1} - t_{k}} \geq 1$
\al{
    p^{(t_{k + 1})}_M(\gamma) &\geq \left(\frac{\delta}{1 + \mu}\right)p^{(t_k)}_M(\gamma).
}
Therefore, when $k\to \infty$, $p^{(t_{k})}_M(\gamma)$ diverges and is not a probability density function, hence the sequence converges uniquely to $p^{*}_M$. As a result, 
\al{
    \ell_{n, \mu}[P_M^{(t)}] &\stackrel{t \to \infty}{\longrightarrow} \ell_{n, \mu}[P_M^{*}]
}

Furthermore, if we consider any sub-sequence of $\{P_{M}^{(t)}\}_{t = 0}^\infty$, $\{P_{M}^{(t_k)}\}_{k = 0}^\infty$, then the sub-sequence must have a further sub-sub sequence which converges to $P_M^*$ weakly by the same argument as above.  Thus, since every sub-sequence has a sub-sub sequence which converges to the maximum the primary sequence $\{P_{M}^{(t)}\}_{t = 0}^\infty$ must converge weakly to $P_M^*$.
\end{proof}

\section{Computational Details} \label{supp:sec:computational_details}

We first elaborate on the details of the nonparametric EM algorithm. 
Recall that the updates in the nonparametric EM algorithm is:
\al{
    p^{(j + 1)}(\gamma|x) &= \sum_y\frac{\hat p(y)p_A(y|\gamma)p_M^{(j)}(\gamma|x)}{p^{(j)}_{MA}(y)}.
}

Although we have a closed form for this updated density, it cannot be succinctly be summarized by set of parameters.  Instead, we consider for each $x$ of interest, sampling $J$ times from the updated distribution and seleting a set of $K$ uniformly spaced quantiles of the distribution. The resultant approximated $p^{(j + 1)}(\gamma|x)$ will consist of a mixture of uniform densities, placed between the estimated quantiles. A similar idea has been proposed for non-parametric conditional independence testing \citep{Petersen2020TestingCopulas}. 

Therefore, if we let $F^{(j + 1)}(\gamma|x)$ to be the conditional CDF corresponding to $p_M^{(j + 1)}(\gamma|x)$. Then  consider a set of quantiles to be simultaneously estimated $\mathcal{T} = \{\tau_1, \dots, \tau_M\}$, consider a quantile regression procedure which computes estimates for $q_{k}(x) = \mathcal{Q}_{p^{(j + 1)}(\gamma|x)}^{\tau_k}(\gamma|x)$ the true conditional quantile functions.  Denote the linear approximation to the conditional distribution function $F(\gamma|x)$

\begin{equation}
F_{L}(\gamma|x) = \begin{cases}
&\gamma\frac{\tau_1}{ q_{1}(x)},\quad \mbox{if $\gamma\leq  q_{1}(x)$}, \\
&\tau_j+ (\tau_{j+1}-\tau_j)\frac{\gamma- q_{j}(x)}{ q_{j + 1}(x)- q_{j}(x)},\quad \mbox{if $ q_{j}(x) < \gamma\leq  q_{j + 1}(x)$}, \\
&\tau_M + (1-\tau_M)\frac{\gamma- q_{M}(x)}{1- q_{M}(x)},\quad \mbox{if $q_{M}(x)<\gamma$}.
\end{cases}
\end{equation}

And from each of the points of interest, we can estimate the sample quantiles to obtain a set of functions $\{\hat q_{k}(x)\}_{k = 1}^M$ which allow us to define the estimated linearization of the cdf. 

\begin{equation} \label{eq:linear_cdf}
\hat F_{L}(\gamma|x) = \begin{cases}
&\gamma\frac{\tau_1}{ \hat q_{1}(x)},\quad \mbox{if $\gamma\leq  \hat q_{1}(x)$}, \\
&\tau_j+ (\tau_{j+1}-\tau_j)\frac{\gamma- \hat q_{j}(x)}{ \hat q_{j + 1}(x)- \hat q_{j}(x)}, \quad \mbox{if $ \hat q_{j}(x) < \gamma\leq  \hat q_{j + 1}(x)$}, \\
&\tau_M + (1-\tau_M)\frac{\gamma- \hat q_{M}(x)}{1- \hat q_{M}(x)},\quad \mbox{if $\hat q_{M}(x)<\gamma$}.
\end{cases}
\end{equation}

Since we are able to sample from $p^{(j+1)}(\gamma|x)$ arbitrarily many times, it is clear that we can approximate $F_{L}(\gamma|x)$ arbitrarily well at each step.  This linear interpolation of the quantiles additionally acts as a smoothing step in the EM algorithm.

\subsection*{Nonparametric EM sampling}
All of the nonparametric EM algorithms rely on sampling from: 

\al{
    p^{(t + 1)}_M(\gamma) &= \sum_{y} p^{(t)}_{MA}(\gamma|y)\hat p(y)
}
At any step $ p^{(t)}_M(\gamma)$ is approximated by a mixture of uniform distributions, specified by the locations of the quantiles. In order to sample from $p^{(t + 1)}_M(\gamma)$ we do the following:

\begin{algorithm}[H] 
  1. Set $n_{accept} = 0$ \\
 \While{ $n_{accept} < J$}{
  a. Sample $\tilde Y^{'}$ from $\hat p(y)$ \\
  b. Define $S = \text{True}$ \\
  \While{ $S$ }{
    i. Sample $\tilde \gamma, \tilde Y$ from $p_{MA}(y, \gamma) = p_A(y|\gamma)p^{(t)}_M(\gamma)$ \\
    ii. \If{$\tilde Y = \tilde Y^{'}$}{
        Accept $\tilde \gamma$ \\
        Set $S = \text{False}$
    }   \Else{
        Reject $\tilde \gamma$
    }
    
  }
 }
 \caption{EM algorithm sampling} 
 \label{alg:EM_sampling}
\end{algorithm}

A brief justification is shown below. 
\begin{proof}
    Let $f(\gamma^*)$ be the density of the accepted $\tilde \gamma$
    \al{
         f(\gamma^*) &= f(\tilde \gamma| \text{ accept } \tilde \gamma) \\
         &= \sum_{y'} \hat p(y') f(\tilde \gamma| \tilde Y^{'} = y',  \tilde Y = y') \\
         \text{By independence } f(\tilde \gamma| \tilde Y^{'} = y',  \tilde Y = y') & = f(\tilde \gamma| \tilde Y = y') \\
         &= p_{MA}(\tilde \gamma|y') \\
         \implies f(\gamma^*) &= \sum_{y'} \hat p(y')p_{MA}(\gamma|y')
    }
    
\end{proof}

We approximate $p_M^{(t + 1)}$ at each step by estimating its quantiles from a set $\mathcal{T}$ then use a mixture of uniforms where the outer boundaries are defined by the estimated quantiles of $p_M^{(t + 1)}$.  We can numerically compute the likelihood by sampling from $p_M$ to estimate: $p_{MA}$.  Since we are approximating the EM step, exact computation is not possible, it is possible for the likelihood to decrease.  We set a threshold so that if the likelihood does not increase by a minimum amount, or decreases then the algorithm terminates. 


\section{Additional simulations}
In this supplement we begin with a comparison of our two algorithms and show how the geometric program is much faster than the nonparametric EM for computing a mixing density.  Additionally we also illustrate the effect of smoothing of the estimated mixing density. 

\subsection{Comparison: NPEM vs GP} \label{supp:sec:speed_test}
We illustrate the difference in performance of the two choices for the latent distribution, the nonparametric EM (NPEM) algorithm and our binned geometric program (GP). We consider a version of the NPEM in which we sample $1000$ times at each iteration which we use to estimate the quantiles levels $\tau \in \{\frac{i}{1000}\}_{i = 1}^{1000}$.  We compare this to a binned approximation where we discretize the latent distribution using $1000$ bins. 

We consider a set of models where $h \in [0.2, 8]$ and the correct measurement model is known. We sample datasets of size $n = 100$ and fit the likelihood for different regularization parameters $\mu$.  
\al{
    \gamma_i &\sim \text{Beta}(12,5) \\
    Y_{i}|\gamma_i &\stackrel{iid}{\sim} MKM(K = \text{Gaussian}, h = h, \gamma = \gamma_i) 
}

We compare the mean time per fit of each method as well as compute the likelihood value for each.  We plot the mean difference in the likelihood as a function of the bandwidth.

\begin{figure}[htp!]
\centering
\includegraphics[height = 0.33\textheight]{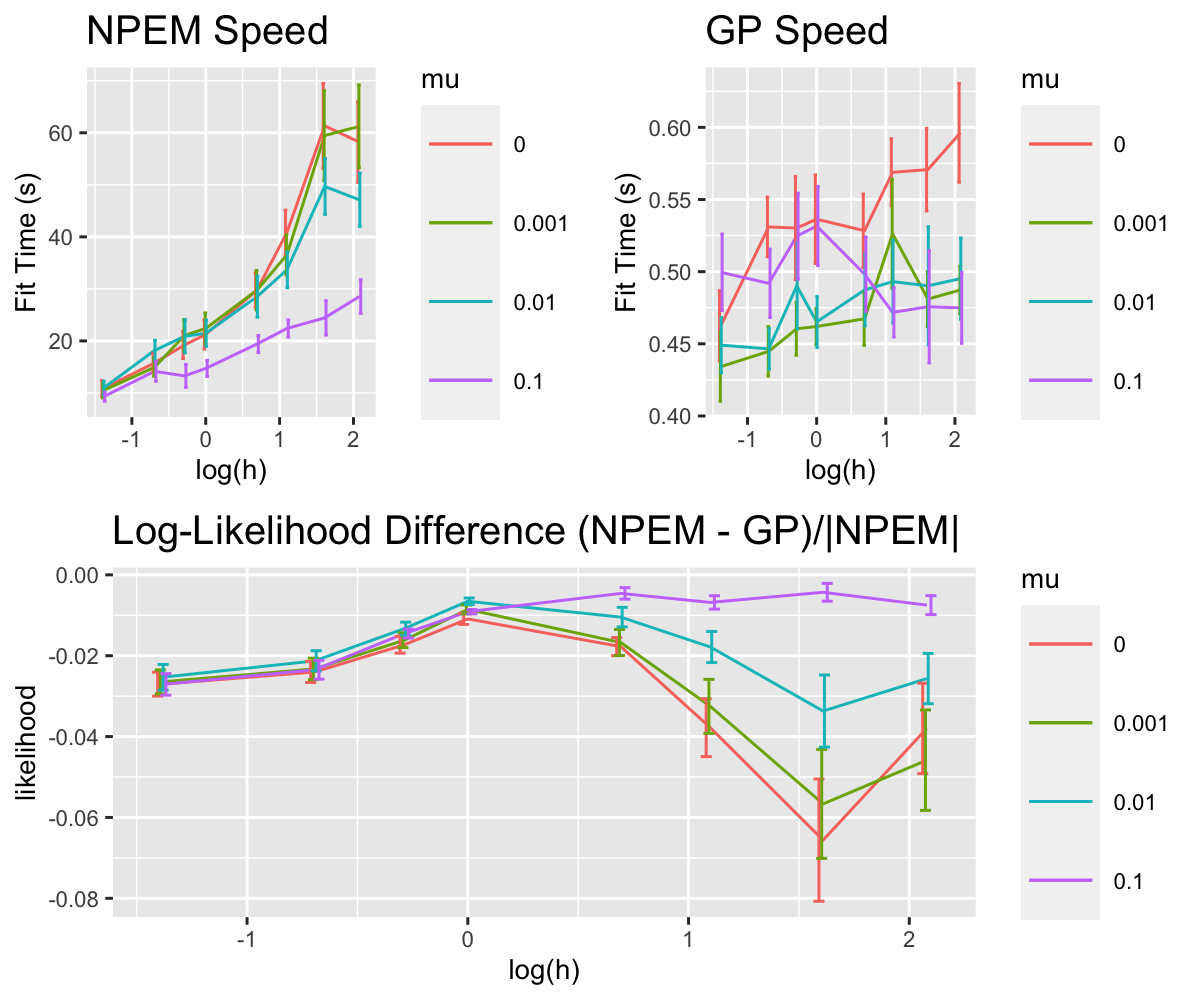}
\caption{Speed difference between the NPEM and GP methods. Note that the relative difference in likelihood for the NPEM algorithm is approximately 0.3\% to 1.5\% less than the Geometric Program, which can compute the estimate faster.}
\label{fig:SpeedTest}
\end{figure}

From Figure \ref{fig:SpeedTest} it is clear that under this setting, the GP algorithm tends to obtain a larger likelihood value and is faster in practice. Additionally, we find that in both algorithms, the time required for optimization increases as $h$ increases.  This is more clear in the NPEM, which agrees with our results from Theorem 3.1 in supplement A.

\subsection{Recovering $p(y)$ and $p_M(\gamma)$}

In this section, we examine the effect of regularization on the estimated latent density.  We simulate a dataset sample of size 100 according to the following true latent distribution and measurement model:
\al{
    \gamma_i &\sim \text{Beta}(12,5) \\
    Y_i|\gamma_i &\sim MKM(K = \text{Gaussian}, h = 2, \gamma = \gamma_i)
}

We illustrate the effect of regularization in the figure below. Additionally we used a discrete approximation of the latent distribution with $1000$ bins. 

\begin{figure}[htp!]
\centering
\includegraphics[height = 0.33\textheight]{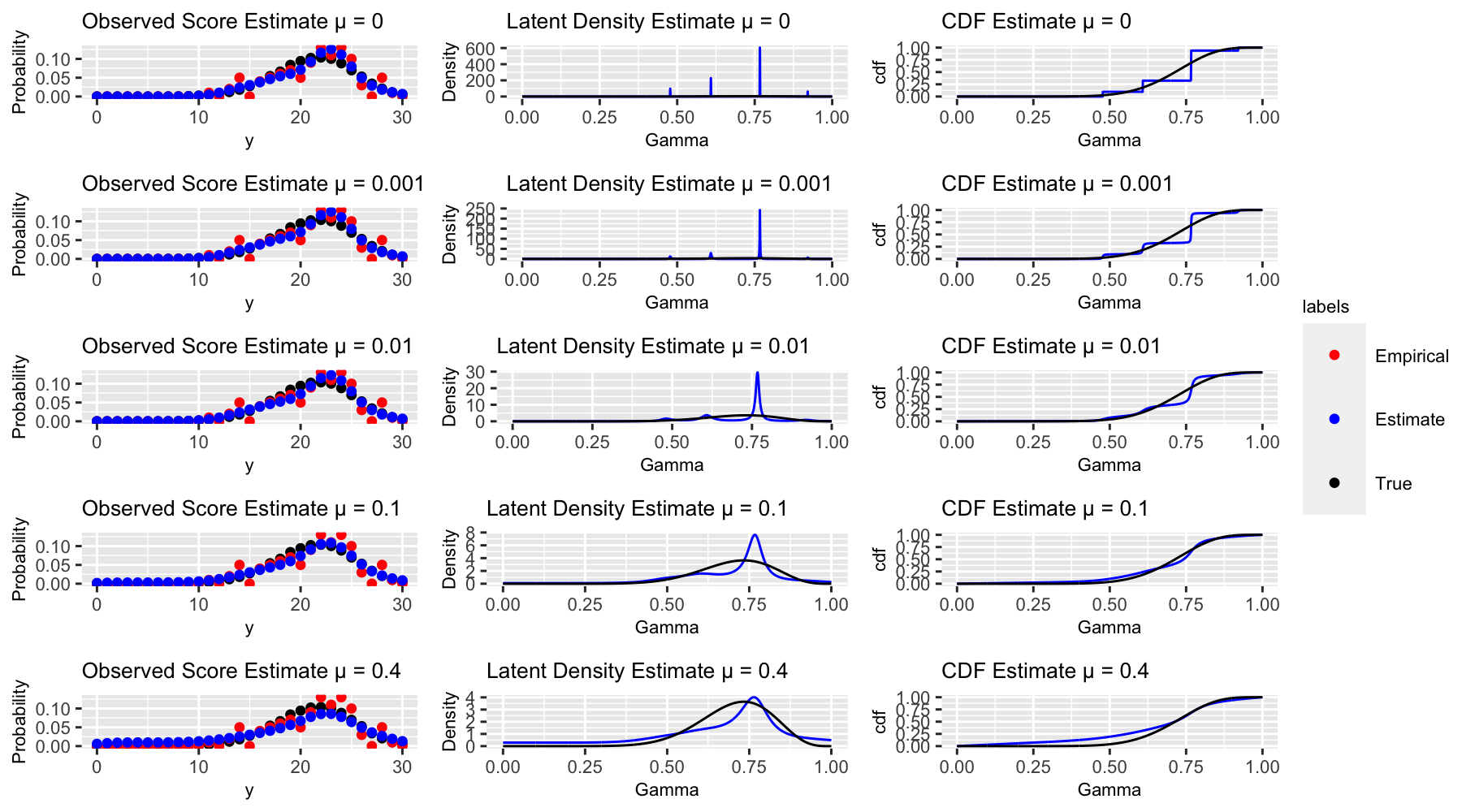}
\caption{Example of the effect of the regularization parameter $\mu$.}
\label{fig:latent_example}
\end{figure}

As we see, in Figure \ref{fig:latent_example} the case of no regularization, when the empirical distribution $\hat p$ is outside of $\text{conv}(\Gamma)$, the NPMLE consists only of point masses, as the closest point $\hat p_{MA}$ exists on the boundary of $\text{conv}(\Gamma)$. However, as we increase $\mu$, the latent distribution becomes smoother in the neighbourhood of each of these masses.  

\section{Additional analysis on the NACC data} \label{supp:sec:nacc_analysis}
In this section we apply the selection methods for $p_A$ and $\mu$ to the NACC dataset. 

As outlined in the main paper we first plot the intrinsic variability $\norm{\hat q - \hat q_{(A)}}_{TV}$  as a function of the measurement assumption model for both the MOCA and the MMSE scores.   From figures \ref{fig:MMSE_intrinsic},\ref{fig:MOCA_intrinsic}
the we pick the binomial model for the MMSE and the Laplace $h = 1.34$ model for the MOCA model, though we also note the Gaussian and binomial models achieved very similar distance to the intrinsic variability distributions (for particular choices of $h$).  We find each of the selected models have p-values of 1 in the first and second order feasibility tests. 

\begin{figure}[htb!]
    \centering
    \includegraphics[height = 0.33\textheight]{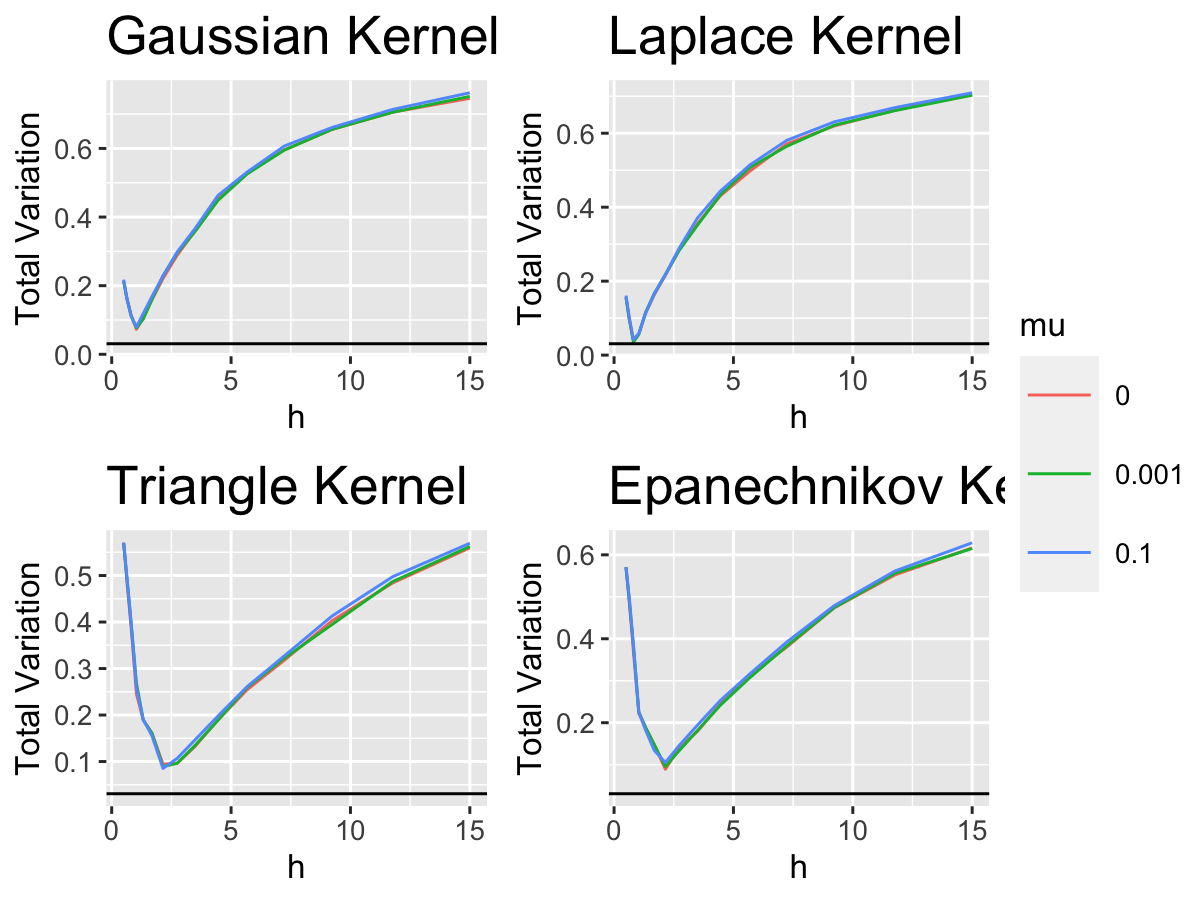}
    \caption{Intrinsic Variability of MMSE.  The binomial model for the measurement assumption is shown in black.}
    \label{fig:MMSE_intrinsic}
\end{figure}

\begin{figure}[htb!]
    \centering
    \includegraphics[height = 0.33\textheight]{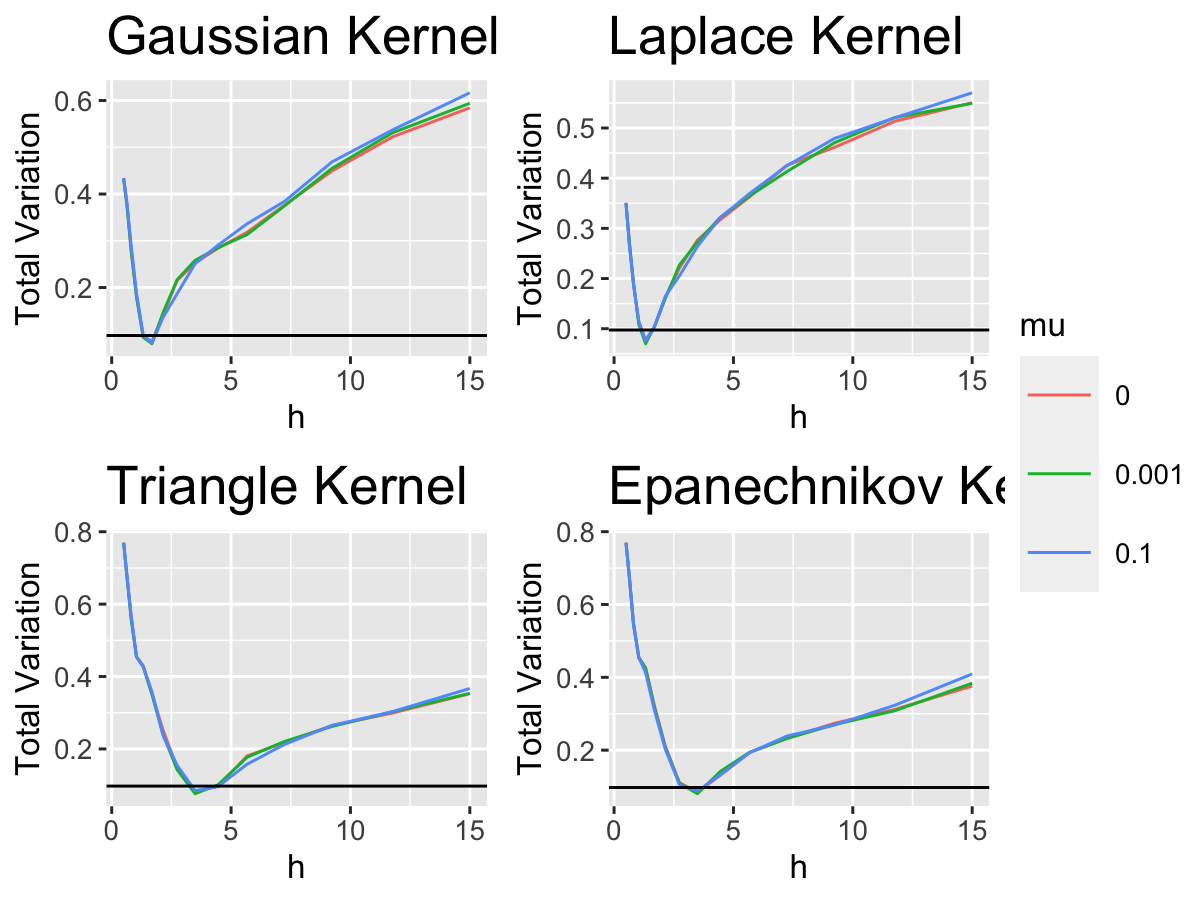}
    \caption{Intrinsic Variability of MOCA.  The Laplace $h = 1.34$ model is optimal. The binomial model for the measurement assumption is shown in black.}
    \label{fig:MOCA_intrinsic}
\end{figure}

\newpage
Next as outlined in the main paper as we choose the regularization parameters for each branch of the model, $\mu_Y$ and $\mu_Z$, corresponding to the MMSE and MOCA latent distributions.  These are chosen using the procedure outlined in the main paper, here we plot the 2-observation likelihood as a function of the regularization parameter. 

From our procedure in figures \ref{fig:MMSE_mu},\ref{fig:MOCA_mu} we select $\mu_Y = 0.0055, \mu_Z = 0.0113$ for our model.

\begin{figure}[htb!]
    \centering
    \includegraphics[height = 0.33\textheight]{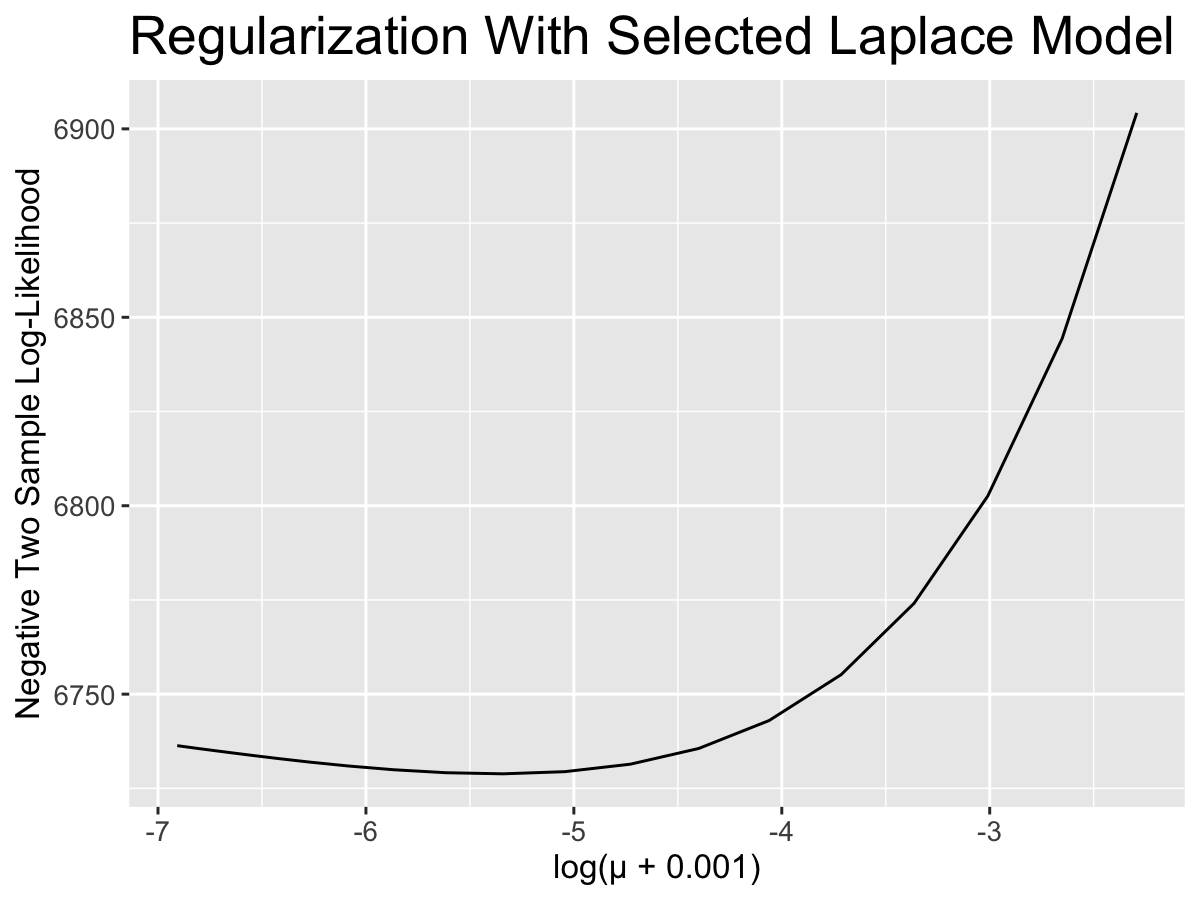}
    \caption{Regularization selection for MMSE model.  $\log$ in the plot is the natural logarithm.  Optimal selected $\mu_Y = 0.0055$. }
    \label{fig:MMSE_mu}
\end{figure}

\begin{figure}[htb!]
    \centering
    \includegraphics[height = 0.33\textheight]{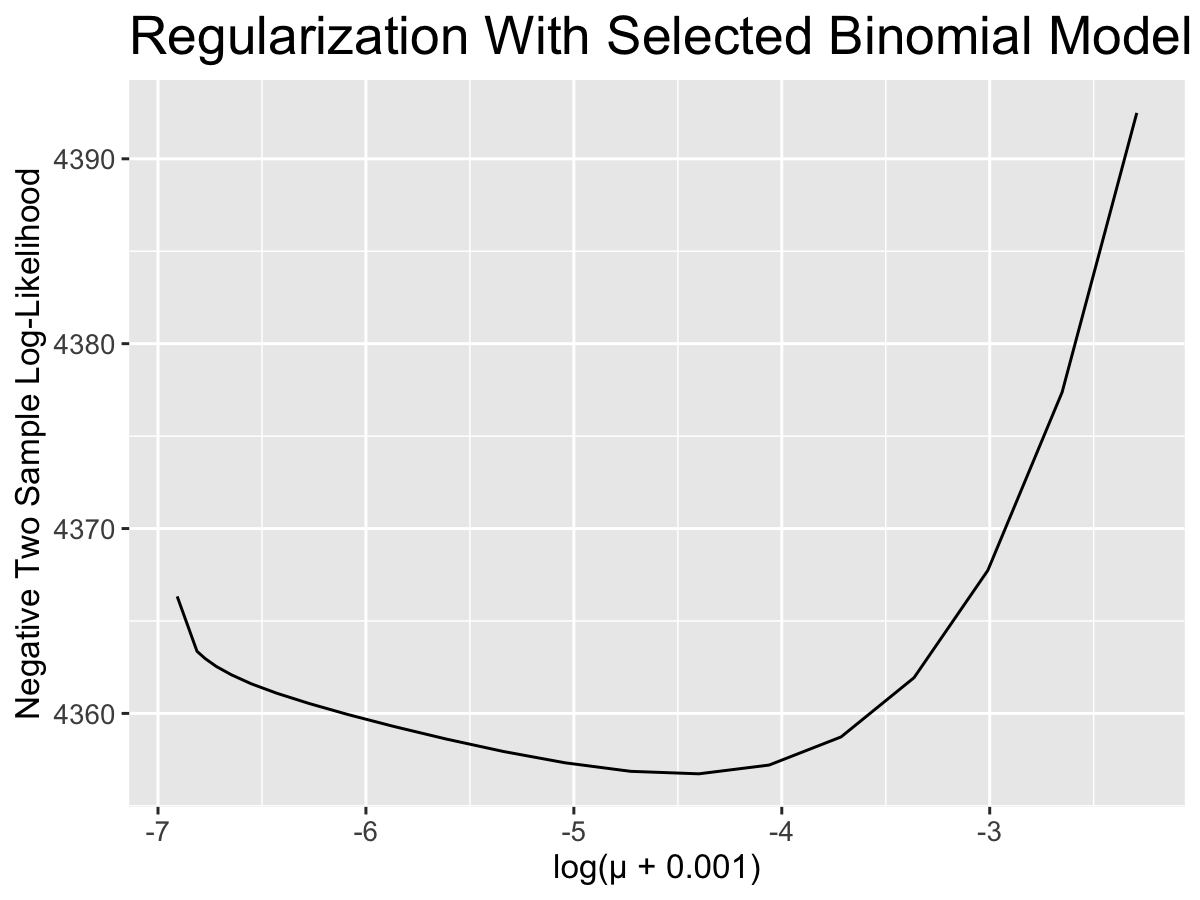}
    \caption{Regularization selection for MOCA model.  Optimal selected $\mu_Z = 0.0113$.}
    \label{fig:MOCA_mu}
\end{figure}

Lastly, we computed the cross entropy on the crosswalk (fully observed) dataset.  Once again, we fix the selected measurement model $p_A$ on $Y$ and plot the cross entropy as a function of the conditional model or alternative method. The parametric model is the logit-normal model introduced in section \ref{sec:par_model}, the Z-score matching method is a currently used method for score conversion introduced in the main paper. As seen in the main paper, the overall optimal model is shown with the red x and the model we select without viewing any complete cases is shown in blue. We can observe that the smoothed non-parametric tends to produce the best conversions.  Conversions using a compact support kernel can perform poorly if they don't assign any probability to an observed value, however, we see that this can be partially combated through regularization.  The parametric model tends to also perform well, though not under the selected $p_A$.  This may be partially explained by the tendency of variational inference methods to underestimate variances \cite{Bishop2004PatternBishop}.  Noticing the improved performance of the parametric model under a larger $h$ seems to support this idea. 
\begin{figure}[htb!]
    \centering
    \includegraphics[height = 0.5\textheight]{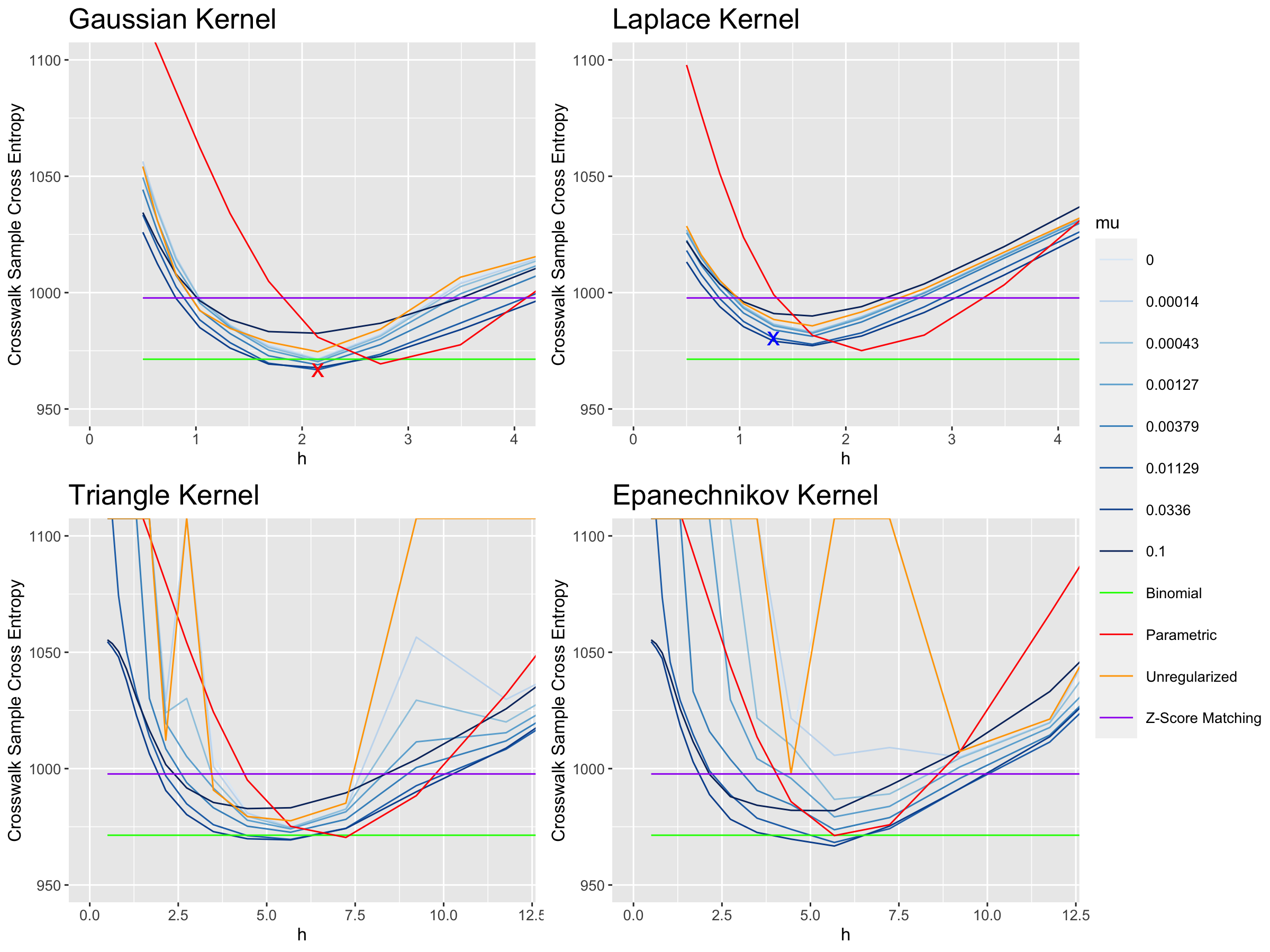}
    \caption{Additional conversion cross entropy plots for other kernel and bandwith combinations for converting MMSE to MOCA.  Red x denotes the lowest cross entropy while blue x denotes the model selected by our model. Note the differences in x axes. }
    \label{fig:conversion}
\end{figure}

\end{document}